\documentclass[conference]{asl}
\usepackage{rotate}
\usepackage{multirow}
\usepackage{bussproofs}
\usepackage{amssymb}
\usepackage{latexsym}
\usepackage{amssymb}
\usepackage{mathrsfs}
\usepackage{graphics}
\usepackage{fancyhdr}
\newtheorem{theorem}{Theorem}

\newtheorem{proposition}[theorem]{Proposition}
\newtheorem{definition}[theorem]{Definition}

\newtheorem{lemma}[theorem]{Lemma}
\newtheorem{corollary}[theorem]{Corollary}
\newtheorem{conjecture}[theorem]{Conjecture}

\numberwithin{theorem}{section}

\usepackage[all]{xy}

\usepackage{multirow}
\usepackage{bussproofs}
\usepackage{amssymb}
\usepackage{latexsym}

\usepackage{tikz-cd}
\usepackage{ragged2e}

\def\copyrightyear{2020}

\title[Higher-order Kripke Models]{Higher-order Kripke models for intuitionistic and non-classical modal logics}

\author[Victor Barroso-Nascimento]{Victor Barroso-Nascimento \\ \medskip University College London}

\address{} {}
{}

\email{victorluisbn@gmail.com}

\urladdr{}

\thanks{}

\begin{document}

\begin{abstract}

This paper introduces higher-order (``nested") Kripke models, a generalization of Kripke models that is remarkably close to Kripke's original idea -- both mathematically and conceptually. Standard models are now $0$-ary models, whereas $n$-ary models for $n > 0$ are models whose set of objects (``possible worlds'') contain only $(n-1)$-ary models. A key idea is the use of worlds as fixed points for modal definitions, in the sense that what is necessary or possible in a world of a frame depends only on what is true in the same world on the accessible frames. This paper mainly deals with the paradigmatic cases of intuitionistic modal logics $IK$ and $MK$, from which the generalisation to other non-classical logics arises naturally. The association between conditions on accessibility relations and modal axioms also carries over to this framework, so modal logics stronger than $K$ can be obtained by imposing requirements on the relations between frames.   Just like Kripke models define a concept of ``alternative'' for classical models, the $n$-ary models (for $n > 0$) defines the same concept for any interpretation of the $(n-1)$-ary models.

\end{abstract}

\maketitle

%The new formulation also sheds some light on a duality between clauses for direct and indirect assertion of connectives.

\section{Introduction}

Classical modal logic extends classical propositional logic by allowing application of the modal notions of necessity and possibility to propositions \cite{sep-logic-modal}.  Of particular importance for it is Saul Kripke's ``possible worlds'' semantics \cite{Kripke1959-KRIACT}, in which modal reasoning takes the shape of counterfactual exploration of what would and would not hold if the facts which are true in the world were different. This framework allows judgments of shape ``A is necessary'' ($\Box A$) to be characterised as true whenever $A$ is true in all alternative configurations of the world, and judgments of shape ``A is possible'' ($\Diamond A$) as true whenever $A$ is true in at least one alternative configuration. The semantics can be further enriched by incorporating a notion of conditional possibility that uses ``acessibility relations'' to specify which worlds are alternative configurations of which other worlds, yielding a plethora of distinct modal logics suitable for different concepts of modality \cite[pgs. 61-112]{Chellas1980}.

There is nothing essentially classical about modalities, so it is natural to wonder how one could obtain modal extensions of non-classical propositional logics. For logics characterizable through simple but non-standard truth functions\footnote{\textit{E.g.} multivalued functions \cite{sep-logic-manyvalued}, ``glutty'' bivalued functions allowing propositions to be both true and false, and  ``gappy'' bivalued functions allowing them to be neither true nor false \cite{FDEAndersonBelnap}\cite{RestallGapsGluts}.} the answer is straightforward, as it suffices to change how truth is treated at possible worlds and adapt modalities accordingly \cite{PriestManyvaluedmodal}. The answer is not as clear in the case of non-classical logics requiring possible world semantics for characterization of their \textit{propositional} fragment, as the concepts employed to provide modal extensions are already used to define the basic semantic notions. This is the case, for instance, of intuitionistic logic, whose main semantic characterization is given in terms of models structurally similar to those of the classical modal logic $S4$ \cite{KripkeIntuitionistic}.

%Much less clear is what answer should be given in the case of non-classical logics that already require possible world semantics for the characterization of their \textit{propositional} fragment, as the concepts usually employed to provide modal extensions are already used to define even their most basic semantic notions. This is the case, for instance, of intuitionistic logic, whose main semantic characterization is given in terms of Kripke models structurally similar to those of the classical modal logic $S4$ \cite{KripkeIntuitionistic}.

This paper provides a conceptually and technically sound answer to this question by showing a general way of obtaining modal semantics for logics characterizable through Kripke models\footnote{Which incidentally also works for Tarskian models, as they can be viewed as special cases of Kripke models.}. Just like it is possible to define classical modal semantics by using a set of classical models (truth-functions, which are the ``possible worlds'') and relations between them, we show that it is possible to define non-classical modal semantics in general by using a set of non-classical models (Kripke models, now regarded as ``possible structures'') and relations between them. In other words, we show that it is possible to extend to modal versions of such non-classical logics essentially the same treatment given to classical modal logic.

This very general framework is arrived at after we discuss how modalities are treated in the paradigmatic case of intuitionistic modal logics, specifically the modal logic $IK$ and a new logic we call $MK$. Those are taken to offer a perfect starting point -- both due to the existence of subtle differences between their semantics, which highlight important features of our new framework, and to the broad relevance of intuitionistic modalities for both philosophy of mathematics and computer science\footnote{A thorough presentation of the philosophical importance of mathematical constructivism is given in \cite{sep-mathematics-constructive}. An equally thorough overview of many possible applications of intuitionistic modalities to computer science is given in section II of \cite{degroot2025semanticalanalysisintuitionisticmodal}.}.

The paper is structured as follows. First, we present the traditional semantic treatment of intuitionistic modalities, which takes the shape of Kripke models with two accessibility relations (birelational models). Then we provide mathematical and conceptual reasons for a radical change in perspective. From these reasons arises a new, intuitive, conceptually robust semantic framework for intuitionistic modalities, which is proven to be equivalent to birelational semantics through effective mappings that preserve model validity. It is then show that the structure of the new models naturally leads to a generalization of the very notion of Kripke models, as well as to modular definitions which yield a modal semantic extension for every logic characterisable through Kripke models. In the case of logics that do not have canonical modal extensions (\textit{viz} for each classical modal logic there is more than one viable candidate for the position of ``true'' non-classical analogue of that logic), the approach always yields a significantly strong modal extension, but weaker extensions may also be obtained by making adjustments. The paper ends with a general discussion on features of the framework and on some of its variants, which includes the statement and justification of important conjectures concerning its properties.

With the exception of short proofs and the proof of Lemma \ref{lemma:modalexistence}\footnote{This is the lemma ensuring the existence of accessible worlds in the canonical model. Lemmas of this nature are usually the most important steps in modal completeness proofs, ours being no exception.}, whose structure is considered interesting in itself, all proofs are in the Appendix.

\section{Propositional intuitionistic logic}

The semantics of intuitionistic propositional logic can be characterized as follows:

\begin{definition}
    A \textit{Kripke frame} is any sequence $\langle W, \leq \rangle$ such that:

\begin{enumerate}
  \item $W$ is a non-empty set of objects;

       \item $\leq$ is a reflexive and transitive relation on $W$;
\end{enumerate}

\end{definition}

\begin{definition}\label{minimalpropositionalmodels}
    A \textit{propositional Kripke model} (or just \textit{propositional model}) $K$ is any sequence $\langle W, \leq, v \rangle$ such that:

   \begin{enumerate}
       \item $\langle W, \leq \rangle$ is a Kripke frame;

         \smallskip
      
       \item $v$ is a function assigning a set $v(w)$ of atomic propositions to each $w \in W$, with the condition that if $w \leq w'$ then $v(w) \subseteq v(w')$.
   \end{enumerate}
\end{definition}

\begin{definition}\label{def:validityinpropositional}
 Given a propositional model $K =\langle W, \leq, v \rangle$ and for $w \in W$, the relations $(\Vdash^{K}_{w})$, $(\Vdash^{K})$ and $(\Vdash_{P})$ are defined as follows\footnote{As usual, crossed relational symbols such as $\nVdash$ are used as an abbreviation of the statement that $\Vdash$ does not hold.}:

 \begin{enumerate}

    \item $\Vdash_{w}^{K} p  \Longleftrightarrow  \ p \in v(w)$, for atomic $p$;

    \smallskip
    
    \item $\Vdash_{w}^{K} A \land B \Longleftrightarrow \ \Vdash_{w}^{K} A$ and $\Vdash_{w}^{K} B$;

      \smallskip
    
    \item $\Vdash_{w}^{K} A \lor B \Longleftrightarrow \ \Vdash_{w}^{K} A$ or $\Vdash_{w}^{K} B$;

      \smallskip
    
    \item $\Vdash_{w}^{K} A \to B \Longleftrightarrow A \Vdash_{w}^{K} B$;
     
        \smallskip
        
    \item $\nVdash_{w}^{K} \bot$ for any $w \in W$;
    
  \smallskip

    \item For non-empty $\Gamma$, $\Gamma \Vdash_{w}^{K} A \Longleftrightarrow \forall w' (w \leq w'): \ \Vdash_{w'}^{K} C$ for all $C \in \Gamma$ implies $ \Vdash_{w'}^{K} A$;

      \smallskip
    
     \item $\Gamma \Vdash^{K} A \Longleftrightarrow \forall w (w \in W): \Gamma \Vdash_{w}^{K} A$.

       \smallskip

    \item $\Gamma \Vdash_{P} A \Longleftrightarrow \Gamma \Vdash^{K} A$ holds for all $K$.
\end{enumerate}

\end{definition}

\begin{definition}
    $\neg A$ is an abbreviation of $A \to \bot$.
\end{definition}

One way of conceptually justifying those definitions is by saying they represent the activities of a mathematician throughout time, as well as the ensuing process of knowledge accumulation. This intuitive interpretation, given by Kripke in \cite[pgs. 97-105]{KripkeIntuitionistic}, is concisely fleshed out by Priest in \cite[pg. 106]{Priest} in a passage we quote (adapting the notation to match ours):

\begin{quote}
    6.3.6 Before we complete the definition of validity, let us see how an
intuitionist interpretation arguably captures the intuitionist ideas of the
previous section. Think of a world as a state of information at a certain time;
intuitively, the things that hold at it are those things which are proved at
this time. $w \leq w'$ is thought of as meaning that $w'$ is a possible extension of $w$,
obtained by finding some number (possibly zero) of further proofs. Given
this understanding, $\leq$ is clearly reflexive and transitive. [...] And the heredity condition is also
intuitively correct. If something is proved, it stays proved, whatever else
we prove.

 6.3.7 Given the provability conditions of 6.2.7, the recursive conditions of
6.3.4 are also very natural. $A \land B$ is proved at a time iff $A$ is proved at that
time, and so is $B$; $A \lor B$ is proved at a time iff $A$ is proved at that time, or
$B$ is. If $\neg A$ is proved at some time, then we have a proof that there is no
proof of $A$. Hence, $A$ will be proved at no possible later time. Conversely, if
$\neg A$ is not proved at some time, then it is at least possible that a proof of
$A$ will turn up, so $A$ will hold at some possible future time. Finally, if $A \to B$
is proved at a time, then we have a construction that can be applied to any
proof of $A$ to give a proof of $B$. Hence, at any future possible time, either
there is no proof of $A$, or, if there is, this gives us a proof of $B$. Conversely,
if $A \to B$ is not proved at a time, then it is at least possible that at a future
time, $A$ will be proved, and $B$ will not be. That is, $A$ holds and $B$ fails at some
possible future time.
\end{quote}

It is a matter of controversy if such a justification is acceptable, or even if the use of Kripke models is intuitionistically justifiable. Some philosophers argue that this reading is misleading \cite{AllenHazen}, whilst some argue that the characterization of intuitionistic semantics through Kripke models itself is misleading \cite{WagnerSanz}. In any case, we will not try to settle such discussions here. The interpretation will be used to better illustrate some key conceptual points in section 5, but the arguments we present are entirely independent of which conceptual reading of the models one adopts.

%In fact, although there are semantics for intuitionistic logic whose completeness proofs are constructive (e.g. Sandqvist's base-extension semantics \cite{Sandqvist2015IL}, a proof-theoretic semantics recently extended to first-order logic by Gheorghiu in \cite{gheorghiu2024prooftheoreticsemanticsfirstorderlogic}), a strong negative result by Kreisel \cite{Kreisel1958-KREMSO} building on previous results by Gödel \cite{Gödel1931} show that completeness for the standard Kripke semantics of intuitionistic predicate logic implies validity of the intuitionistically unnaceptable Markov's Principle in the metalanguage \cite{MccartyCompletenessIncompleteenss}. This cannot be taken as a general feature of intuitionistic Kripke semantics, however, since by allowing ``explosive worlds'' in which every formula is valid and changing the semantics of $\bot$ so that it is only satisfied in such worlds instead of being unsatisfiable (as also done in the aforementioned proof-theoretic semantics by Sandqvist and Gheorghiu) it becomes possible to give a constructive proof of completeness \cite{IlikHeberClassicalKripkeSemantics}\cite{VeldmanIntuiConstructiveCompleteness}. In any case, we will not try to settle here the discussion on whether the use of Kripke models or the interpretation outlined above are acceptable. The interpretation will be used to better illustrate some key conceptual points in section 5, but the arguments we present are entirely independent of which conceptual reading of the models one adopts.

\section{Intuitionistic modal logic and birelational models} 

Unlike classical logic, intuitionistic logic does not satisfy \textit{categoricity}, in the sense that for every classical modal logic there is more than one intuitionistic modal logic that could claim to be its true analogue. An overview of debates and results on this issue is given by Simpson in \cite[pgs. 41-44]{Sim94}. Mathematical criteria for establishing reasonable analogues of classical modal logics have been given by Fischer Servi in a series of papers \cite{Servi1ModalAnalogue}\cite{Servi2finitemodel}\cite{Servi3rare}, as well as independently by Plotkin and Stirling \cite{PlotkinStirling1988-PLOAFF}. It is also possible to list series of \textit{desiderata} that a reasonable intuitionistic modal logic should fulfill in order to provide a philosophical argument for its adequacy, as does Simpson in \cite[pgs. 38-41]{Sim94}. After listing his desiderata, Simpson argues that the true intuitionistic analogue of the classical modal logic $K$ is the intuitionistic logic $IK$ \cite[pgs. 58-64]{Sim94}. $IK$ is also the analogue of $K$ obtained in the frameworks of Fischer Servi and of Plotkin and Stirling.

As mentioned before, since intuitionistic proposional models are already Kripke models, it is not entirely clear how models for intuitionistic modal logics should look like. The traditional answer given by the literature takes the shape of \textit{birelational models} -- which, as the name suggests, are just Kripke models with two accessibility relations. The first relation, $\leq $, induces the desired intuitionistic behaviour on the propositional semantics, while the second, $R$, plays the role of a proper modal relation.

Some practical difficulties arise when we attempt to define intuitionistic semantics this way. The following property can easily be proved for the propositional (also first-order) semantics \cite[pgs. 105-106 and 423]{Priest}, and is essential for ensuring proper intuitionistic behaviour:

\begin{proposition}[Monotonicity]
     If $\Vdash^{K}_{w} A$ and $w \leq w'$ then $\Vdash^{K}_{w'} A$
\end{proposition}

Due to how entailment is defined, the proof can be extended to also show that if $\Gamma \Vdash^{K}_{w} A$ and $w \leq w'$ then $\Gamma \Vdash^{K}_{w'} A$.

If we want properly intuitionistic modalities, we need to ensure that they also satisfy monotonicity. However, not all combinations of intuitive modal clauses with birelational models make modalities monotonic. A short overview of some possible solutions is given by Simpson\footnote{Incidentally, Simpson's thesis gives a comprehensive historical overview of the development of those modal semantics and is itself an authoritative source on the syntax and semantics of intuitionistic modal logics.} in \cite[pgs. 45-57]{Sim94}. They essentially come in two kinds: either monotonicity is directly built into the modal clauses or all models are required to satisfy \textit{frame conditions} which allow monotonicity to be proved even for simpler modal clauses.

The birelational semantics for $IK$ solves the issue by building monotonicity in the clause for $\Box$ and requiring satisfaction of frame conditions, one of which makes $\Diamond$ monotonic. The conditions it requires for the modal relation $R$ and the non-modal $\leq$ are $(F_1)$ and $(F_2)$:

\begin{definition}\label{F1andF2}
 Conditions $F_1$ and $F_2$ are defined as follows:

\begin{enumerate}

     \item[($F_1$)] If $w \leq w'$ and $w R j$ then there is a $j'$ such that $j \leq j'$ and $w' R j'$

        \item[($F_2$)] If $w R j$ and $j \leq j'$ then there is a $w'$ such that $w \leq w'$ and $w' R j'$

\end{enumerate}

\end{definition}

Which can also be expressed diagrammatically as follows, where the solid lines represents the antecedent of the conditions and the dotted lines represent the relations that are required to hold when the antecedent holds:

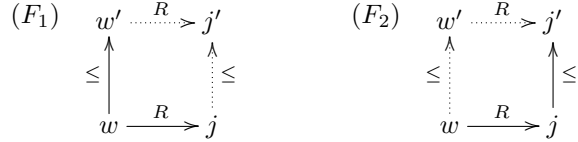
\begin{center}
$(F_1)$\;
\xymatrix{
w' \ar@{.>}[r]^R & j' \\
w \ar[u]^{\le} \ar[r]^R & j \ar@{.>}[u]_{\le}
}
\qquad \qquad
$(F_2)$\;
\xymatrix{
w' \ar@{.>}[r]^R & j' \\
w \ar@{.>}[u]^{\le} \ar[r]^R & j \ar[u]_{\le}
}
\end{center}

Condition $F_1$ makes the simpler definition of $\Diamond$ monotonic.

It is not a requirement that the objects $j'$ and $w'$ are unique; if they are, the model is said to be \textit{universal}. $IK$ is also sound and complete with respect to universal models \cite[pg. 152]{Sim94}. Since some of our proofs rely on this uniqueness, the models we'll define will be universal models.

The natural complements of those two conditions would be the following, which are not required to hold in $IK$:

\begin{definition}\label{def:F3andF4}
 Conditions $F_3$ and $F_4$ are defined as follows:

\begin{enumerate}

     \item[$(F_3)$] if  $w \leq w'$ and $w' R j'$ then there is a $j$ such that $w R j$ and $j \leq j'$

        \item[$(F_4)$] If $j \leq j'$ and $w' R j'$ then there is a $w$ such that $w R j$ and $w \leq w'$

\end{enumerate}

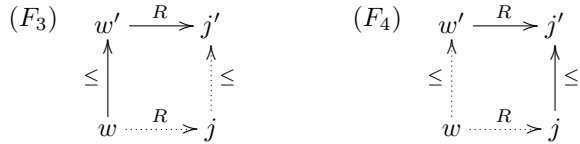
\begin{center}
$(F_3)$\;
\xymatrix{
w' \ar[r]^R & j' \\
w \ar[u]^{\le} \ar@{.>}[r]^R & j \ar@{.>}[u]_{\le}
}
\qquad \qquad
$(F_4)$\;
\xymatrix{
w' \ar[r]^R & j' \\
w \ar@{.>}[u]^{\le} \ar@{.>}[r]^R & j  \ar[u]_{\le}
}
\end{center}

\end{definition}

%A few reasons on why $F_{3}$ and $F_{4}$ may be considered mathematically unnatural, and thus not expected to hold, are given in \cite{PlotkinStirling1988-PLOAFF}, which also gives mathematical justifications for $F_{1}$ and $F_{2}$. 

Conceptual justifications for $F_{1}$ and $F_{2}$ are given in \cite[pg. 51]{Sim94}, and mathematical justifications are given in \cite{PlotkinStirling1988-PLOAFF}. Roughly speaking, $F_{1}$ can be read as preservation of the property of ``accessing a (suitable) world'', and $F_{2}$ as the property of ``being accessed by a (suitable) world''. It is also worth remarking that $F_{3}$ was used by Božić and Došen in \cite{Bozic1984-BOIMFN}\cite{Dosen1985-DOEMFS}.

%establishing interactions between the model's two relations, which then allows modal monotonicity to be proved.

%Universal birelational models are birelational models plus uniqueness requirement for witnesses. \cyan{Obs: tese do Simpson também mostra uma prova de completude para modelos universais.}

Once those basic definitions have been stablished, models for $IK$ can be defined as follows:

\begin{definition}\label{def:birelational}
    A \textit{birelational model} $B$ is any sequence $\langle W, \leq, R , v \rangle$ such that:

   \begin{enumerate}
       \item $W$ is a non-empty set of objects;
       
       \item $\leq$ is a reflexive and transitive relation on $W$;

       \item $R$ is a relation on $W$;
      % \item $\leq$ is a reflexive, transitive and anti-symmetric binary relation on the elements of $W$;
      
       \item $\textit{v}$ is a function assigning a set of atomic propositions $\textit{v}(w)$ to each $w \in W$, satisfying the condition that $w \leq w'$ implies $v(w) \subseteq v(w')$.

       \item Conditions $F_{1}$ and $F_{2}$ are satisfied. The $j'$ of $F_{1}$ and the $w'$ of $F_2$ are also required to be unique.

   \end{enumerate}

\end{definition}

\begin{definition}\label{def:validityinbirelational}
 Given any birelational model $B$ with set of objects $W$ and any $w  \in W$, the relations $(\Vdash^{B}_{w})$, $(\Vdash^{B})$ and $(\Vdash)$ are defined as follows:

 \begin{enumerate}

    \item Clauses 1 through 8 are as in Definition \ref{def:validityinpropositional}, but all occurrences of $K$ are replaced by $B$ and $\Vdash_{P}$ is replaced by $\Vdash$;

    \setcounter{enumi}{8}

    \item $\Vdash_{w}^{B} \Box A \Longleftrightarrow $ for all $w'$, if $w \leq w'$ and $w' R w''$ then $\Vdash_{w''}^{B} A$;

    \smallskip

    \item $\Vdash_{w}^{B} \Diamond A \Longleftrightarrow $ there is a $w'$ such that $ w R w'$ and $\Vdash_{w'}^{B} A$;

\end{enumerate}

\end{definition}

So monotonicity is built into $\Box$, whereas it can easily be proven to hold for $\Diamond$ by using condition $F_{2}$ in a straightforward inductive proof. The use of $\leq$ in the clause for $\Box$ is not seen as problematic because essentially the same treatment is given to intuitionistic implication and universal quantification, which are generally viewed as non-modal analogues of necessity.

Monotonicity does not have to be built into $\Box$ if models are also required to satisfy condition $F_{3}$. In this case, we obtain a semantics for a logic we call $MK$. It follows from our results that by requiring all models to satisfy conditions $F_{1}$, $F_{2}$, $F_{3}$ and $F_{4}$ we still obtain semantics for $MK$, therefore $MK$ is the maximal modal logic obtainable through requirement of satisfaction of those simple frame conditions.

%A corollary of some of our results will be that requiring models to satisfy both $F_{3}$ and $F_{4}$ still yields a semantics for $MK$, therefore $MK$ is the maximal modal logic obtainable through requirement of satisfaction of those simple frame conditions.

Models for $MK$ are defined as follows:

\begin{definition}
    A \textit{strong model} $S$ is a birelational model satisfying condition $F_{3}$ for unique $j$.
\end{definition}

\begin{definition}\label{def:validityinstrong}
 Given any strong model $S$ with set of objects $W$ and any $w  \in W$, the relations $(\Vdash^{S}_{w})$, $(\Vdash^{S})$ and $(\Vdash^{*})$ are defined as follows:

 \begin{enumerate}

    \item Clauses 1 through 8 are as in Definition \ref{def:validityinpropositional}, but all occurrences of $K$ are replaced by $S$ and $\Vdash_{P}$ is replaced by $\Vdash^{*}$;

    \setcounter{enumi}{8} 

    \smallskip

    \item $\Vdash_{w}^{S} \Box A \Longleftrightarrow $ for all $w'$, $wRw'$ implies  $\Vdash_{w'}^{S} A$;

\smallskip

    \item $\Vdash_{w}^{S} \Diamond A \Longleftrightarrow $ there is a $w'$ such that $ w R w'$ and $\Vdash_{w'}^{S} A$;

\end{enumerate}

\end{definition}

We also define the following stronger notion of model for the sole purpose of showing later that it still yield a semantics for $MK$, which can then indeed be taken as maximal with respect to the frame conditions:

\begin{definition}\label{excessivemodels}
    A \textit{excessive model} $E$ is a strong model satisfying condition $F_{4}$ for unique $w$.
\end{definition}

Although this is not immediately obvious, $MK$ is slightly stronger than $IK$, as witnessed by the following results\footnote{We are thankful to Alex Simpson for providing those counterexamples and dispelling our initial belief that this was still a semantics for $IK$.}:

\begin{proposition}\label{prop:differenceIKMKfirst}
    Let $\top$ be a theorem of $MK$. Then the following hold:

\begin{enumerate}
    \item $\Vdash^{*} (\neg \Box \bot) \to (\Diamond \top)$

    \item $\Vdash^{*} (\Box( A \lor \neg A) \land \neg \Box A) \to (\Diamond \neg A)$.
\end{enumerate}
\end{proposition}

\begin{proposition}\label{prop:IKMKdifferencesecondpart}
    Let $\top$ be a theorem of $IK$. Then the following hold:

\begin{enumerate}
    \item $\nVdash (\neg \Box \bot) \to (\Diamond \top)$

    \item $\nVdash (\Box( A \lor \neg A) \land \neg \Box A) \to (\Diamond \neg A)$.
\end{enumerate}
    
\end{proposition}

Since the validities of $MK$ which are not validities of $IK$ are not unreasonable and $MK$ is still intuitionistic (as the excludded middle can be easily disproved), $MK$ seems to be either the strongest or one of the strongest ``reasonable'' intuitionistic modal logics.

%but for the purposes of this paper (specifically for our conceptual discussions) we assume that they are\footnote{None of our results or discussions essentially depend on the semantics of $\bot$, so any philosophically inclined readers that interpret Kreisel's result as a condemnation of the use of unsatisfiability clauses for $\bot$ should keep in mind that our semantics also allow the non-standard treatment in terms of logical explosion.}.

%The main conjecture is that, even though they are remarkably similar to standard Kripke models, Higher-order Kripke models may be capable of providing semantics to logics which are Kripke-incomplete.

\section{A change of perspective} In the traditional approach, Kripke models for classical modal logic are essentially obtained by picking a set of classical propositional models and putting some relation between them. Birelational models for intuitionistic modal logic are obtained by picking a \textit{single} intuitionistic propositional model and putting a second relation \textit{inside} it. A classical model can intuitively be read as a collection of alternative propositional models, but no such reading can be given to birelational models. They are best understood as a single model with alternative \textit{worlds}. Moreover, each world in a model may also be viewed as a classical propositional model, only the Kripke model itself being intuitionistic.

The traditional approach is conceptually justified by the intended meaning of modalities. On the other hand, most justifications for the birelational approach are essentially mathematical in nature. Discussions on what intuitionistic modal definitions should look like often revolve around which conditions would it be natural for the models to satisfy, or which formulas would it be reasonable for them to validate, instead of what should be the intended meaning of the modalities or of the semantic constructs.

%Discussions on what intuitionistic modal definitions should look like often revolve around which conditions would it be natural for the models to satisfy or which formulas would it be reasonable for them to validate instead of what should be the intended meaning of the modalities or of the semantic constructs.

We will now show that it is possible to provide intuitionistic modal semantics through a generalization of the traditional approach. We start by defining an abstract concept of model:

\begin{definition}\label{def:generalmodel}
    An \textit{abstract model} $A$ is a sequence $\langle W, \succ \rangle$, where $W$ is a non-empty set of Kripke models and $\succ$ is a relation $W$.
\end{definition}

The modal accessibility relation $\succ$ is now defined as holding between Kripke models instead of between the worlds of a model. The intuitionistic propositional models themselves are now considered the ``possible worlds'' of the abstract model, and the accessibility relation once again establishes which models are alternative versions of which other models.  

Although the idea behind it is quite intuitive, this definitions creates a practical issue: when defining relations between worlds, we are capable of handpicking a world $w$ and putting it in a modal relation with a specific world $w'$, provided we make sure all required frame conditions are satisfied. On the other hand, when the relation is defined for Kripke models, we have no control over the accessibility relation at the level of worlds. In a sense, by defining an accessibility relation for models we are uniformly defining modal accessibility for all worlds of that model. This is problematic because validity for formulas is first defined at the level of worlds and then carried over to the entire Kripke model, so without some sort of control it seems modalities could not hold in a world without also holding in all others.

The key insight behind the solution to this is that every world in a Kripke model can be taken to conveniently fix a reference point for the definition of its own modal validities  -- namely, \textit{itself}. Let us return to the interpretation of models as timelines in order to illustrate this point. Consider a frame with worlds $\{m, a, e\}$, representing the possible activities of a particular mathematician in the morning, afternoon and evening of a single work day. Now take an abstract model containing three models $K$, $K'$ and $K''$ defined with no other objects. Let us represent the models graphically as follows:

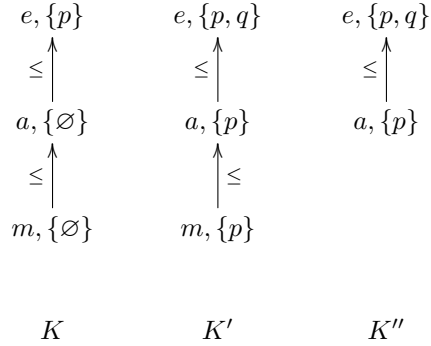
\[
\xymatrix{
e, \{p\} & e, \{p, q\} & e, \{p, q\} \\
a, \{\varnothing \} \ar[u]^{\le} & a, \{p\} \ar[u]^{\le} & a, \{p\} \ar[u]^{\le} \\
m, \{\varnothing \} \ar[u]^{\le} & m, \{p \} \ar[u]_{\le} &  \\
K  & K' & K''
}
\]

Timeline $K$ represents an unproductive day in which, even though the mathematician started working in the morning, they only managed to prove a single theorem $p$ in the evening. In the second timeline, $K'$, the mathematician proves $p$ in the morning (which stays proven afterwards) and manages to prove theorem $q$ in the evening. In the third timeline, $K''$, the mathematician only starts working in the afternoon, but they she still manages to prove theorem $p$ and then prove $q$ in the evening.

Let us assume that, given some criteria of modal accessibility, we have established that $K'$ and $K''$ are acceptable alternative versions of timeline $K$, so $K \succ K'$ and $K' \succ K''$ in our model. How should modal formulas be evaluated at each point of the day? The idea behind the insight is that, in order to evaluate what is possible and necessary in the morning of any given timeline, we only have to check what is and is not the case in the \textit{morning} of its alternative versions. So, for instance, $\Diamond p$ holds in the morning of $K$ because there is at least one timeline which is an alternative to $K$ in which $p$ is proved in the morning (namely, $K'$). In the case of $\Box$ we can either follow the $IK$ definition and also consider the future states of the same time period or the $MK$ modal definition and consider only that specific time period, but in any case we have that $\Box q$ holds in the evening of $K$ if it is not considered an alternative to itself (which is not a given) but does not hold if it is.

The conceptual justification of this idea in no way depends on the reading of models $K$ in terms of timelines. Regardless of what we consider models and worlds to be, the evaluation of a modal formula in a world $w$ of a model $K$ needs only consider what holds in copies of that same world $w$ in models accessible from $K$, which are implicitly taken to be its alternatives. In other words, the ``identity'' of a world is fixed by the object used to represent it in a frame, and modal formulas are evaluated in a world of a model by taking into account its possible occurrences in different propositional models.

Since modal validity in a propositional model depends on which frames are used in other propositional models of the same abstract model, it might be useful to impose some conditions on the frames themselves, which amounts to being selective about which propositional models we allow in our abstract models. In this paper we will deal with abstract models satisfying two kinds of restrictions: \textit{homogeneous models}, in which all models have to be defined using the same frame, and \textit{partially homogeneous models}, in which the frame of each model has to be at least a partial copy of the frame of some other model (satisfying some additional restrictions). The investigation of abstract models themselves will not be conducted here, but we will prove that homogeneous models yield a semantics for $MK$ and partially homogeneous models a semantics for $IK$. Remarkably, no frame conditions or similar structures are necessary. As our proofs show, the frame conditions naturally emerge from interactions between the semantic clauses and the kinds of propositional models we allow in abstract models.

%This is not strictly necessary because, according to the definitions we will provide, when a general model only contains models with disjoint frames the model itself behaves as if there were no accessibility relations at all.

%In this paper we will deal with general models satisfiying two kinds of restrictions: \textit{homogeneous models}, in which all models have to be defined using the same frame, and \textit{partially homogeneous models}, in which the frame of each model has to be at least a partial copy of the frame of some other model (satisfying some additional restrictions). The investigation of general models themselves will not be conducted here, but we will prove that homogeneous models yield a semantics for $MK$ and partially homogeneous models a semantics for $IK$. Remarkably, no frame conditions or similar structural features are necessary. As our proofs show, the frame conditions naturally emerge from interactions between the semantic clauses and the kinds of propositional models we allow in our modal models.

Intuitively, the frame conditions emerge from the new definitions because, when $K \succ K'$, every $ w \in W_{K}$ can be taken as implicitly accessing its own copy in $W_{K'}$. Going back to the example given before, by putting $K \succ K'$ and $K \succ K''$  we would have the following:

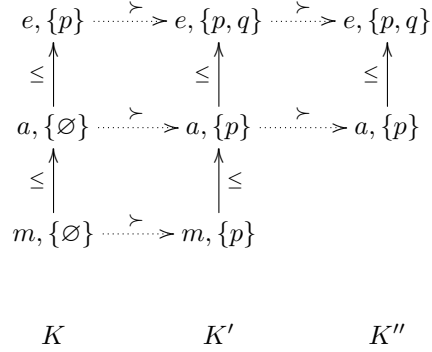
\[
\xymatrix{
e, \{p\} \ar@{.>}[r]^{\succ} & e, \{p, q\}  \ar@{.>}[r]^{\succ} & e, \{p, q\} \\
a, \{\varnothing \} \ar[u]^{\le} \ar@{.>}[r]^{\succ} & a, \{p\} \ar[u]^{\le} \ar@{.>}[r]^{\succ} & a, \{p\} \ar[u]^{\le} \\
m, \{\varnothing \} \ar[u]^{\le} \ar@{.>}[r]^{\succ} & m, \{p \} \ar[u]_{\le} &  \\
K  & K' & K''
}
\]

Hence by imposing requirements on the frames of propositional models we can indirectly induce certain combinations of frame conditions. Notice that this also makes it easier to construct countermodels, as instead of checking whether a model satisfy a given set of frame conditions we only have to make sure that it is defined using the correct frames. For instance, if we want to build a strong model we have to check that the defined structure indeed satisfies $F_{1}$, $F_{2}$ and $F_{3}$, but in order to build a homogeneous model it suffices to just define all propositional models using the same frame.

The idea of using worlds as reference for their own modal definitions is also the main feature that distinguishes our semantics from those present in \cite{vanderGiessen2025-VANOTC-12} and \cite{mojtahedi2026provabilitylogicha}. In those semantics, worlds are associated with (not necessarilly classical) theories instead of Kripke models, meaning that in both semantics we have an accessibility relation between possibly non-classical worlds. On the other hand, theories are syntactic objects lacking the inner structure of Kripke models, meaning that they are not compatible with this idea. As will be seen, the idea allows a very refined control of the semantics through requirements on frames, meaning that we can change the semantics by simply restricting the structure of propositional models. It also allows us to easily avoid undesirable semantic properties, such as validation of the Box Excluded Middle axiom \cite{vanderGiessen2025-VANOTC-12}. Finally, although this will not be explored here, if we use $1$-ary models that relates Kripkean $0$-ary models instead of Kripke models that relates theories, models are allowed be further nested arbitrarily many times.

From now on we denote the set of objects, relation and valuation function of a particular propositional model $K$ respectively as $W_{K}$, $\leq_{K}$ and $\textit{v}_{K}$ whenever it is convenient. The notation will also be used for other sets when the object being referred to in the subscript is clear from the context.

\begin{definition} \label{def:partialcopy}
     Let $K$ and $K'$ be propositional models. Then $K'$ is \textit{at least a partial copy} of $K$ if the following conditions are met:

    \begin{enumerate}
        \item $W_{K'} \subset W_{K}$;

        \item If $w \leq_{K} w'$ and $w \in W_{K'}$ then $w' \in W_{K'}$;

        \item For all $w, w' \in K'$, $w \leq_{K} w'$ iff $w \leq_{K'} w'$;
        
    \end{enumerate}
\end{definition}

This is also called a \textit{generated subframe} in some sources (cf. \cite[pg. 28]{Chagrov1997-CHAML-2}).

In short, an at least partial copy $K'$ of $K$ is just a model obtained by repeating at least part of the objects of $K$ without removing any object $w'$ such that $w \leq_{K'} w'$ for the replicated $w$, as well without promoting changes on the relation $\leq_{K'}$ other than those warranted by the removal of objects. Intuitively, this definition allows us to remove worlds of a frame, but \textit{only} if we keep all ``future states'' of worlds that remain in the frame.

\begin{definition} \label{definition:atleastpartiallyhomogeneous}
     A non-empty set $F$ of propositional models is \textit{at least partially homogeneous} if there is a $K \in F$, called its \textit{reference model}, such that $K' \in F$ implies  $K'$ is at least a partial copy of $K$.
\end{definition}

The following property of at least partially homogeneous models will be useful later:

\begin{lemma} \label{lemma:structureofpartialmodels}
    If $w \leq_{K} w'$, $w \in W_{K'}$ and 
    $K$ and $K'$ are contained in a at least partially homogeneous set $F$, then $w' \in W_{K'}$ and $w \leq_{K'} w'$.
\end{lemma}

\begin{proof}
    Assume $w \leq_{K} w'$ and $w \in W_{K'}$. Let $K''$ be the reference model of $F$. By Clause 1 of Definition \ref{def:partialcopy} we conclude $w, w' \in W_{K''}$, and by Clause 3 we conclude $w \leq_{K''} w'$. But since $w \in W_{K'}$ we conclude $w' \in W_{K'}$ by Clause 2, so by Clause 3 also $w \leq_{K'} w'$.
\end{proof}

\begin{definition} \label{def:homogeneouset}
     Let $F$ be a non-empty set of propositional models. $F$ is said to be \textit{homogeneous} if, for any $K \in F$ and $K' \in F$,  $W_{K} = W_{K'}$ and $\leq_{K} \ = \ \leq_{K'}$.
\end{definition}

 So in a partially homogeneous set all propositional models share part of their frame with a fixed reference model, whereas in a fully homogeneous set all models have the same frame. In both cases models may differ with respect to their atomic assignment functions, but in fully homogeneous sets they may \textit{only} differ with respect to their functions.

\begin{definition}\label{def:partialmodel}
    An \textit{at least partially homogeneous model} (or just \textit{partial model}) $M$ is any sequence $\langle F, \succ \rangle$ such that:

   \begin{enumerate}
       \item $F$ is an at least partially homogeneous set of propositional models;
       
       \item $\succ$ is a relation on $F$.
   \end{enumerate}
\end{definition}

\begin{definition}\label{def:homogeneousmodel}
    A \textit{homogeneous model} $H$ is any sequence $\langle F, \succ \rangle$ such that:

   \begin{enumerate}
       \item $F$ is a homogeneous set of propositional models;
       
       \item $\succ$ is a relation on $F$.
   \end{enumerate}
\end{definition}

%When dealing with multiple propositional models, we may denote by $W_{K}$, $\leq_{K}$ and $v_{k}$ the elements of a propositional model $K$.

%We fix the notation $K_{w}$ for denoting particular propositional models with $w \in W_{K}$, which is useful when dealing with partial modal models.

\begin{definition}\label{def:validityinpartialsimilaritymodel}
 Given any partial model $M$, any $K \in F_{M}$ and any $w \in W_{K}$, the relations $(\vDash^{M, K}_{w})$,  $(\vDash^{M, K})$ $(\vDash^{M})$ and $(\vDash)$ are defined as follows:
 \begin{enumerate}

    \item $\vDash_{w}^{M, K} p  \Longleftrightarrow  \ p \in v_{K}(w)$, for atomic $p$;

    \smallskip
    
    \item $\vDash_{w}^{M, K} A \land B \Longleftrightarrow \ \vDash_{w}^{M, K} A$ and $\vDash_{w}^{M, K} B$;

    \smallskip
    
    \item $\vDash_{w}^{M, K} A \lor B \Longleftrightarrow \ \vDash_{w}^{M, K} A$ or $\vDash_{w}^{M, K} B$;

    \smallskip
    
    \item $\vDash_{w}^{M, K} A \to B \Longleftrightarrow A \vDash_{w}^{M, K} B$;

    \smallskip

    \item $\nvDash_{w}^{M, K} \bot$ for any $w \in W_{K}$ and $K \in F_{M}$;

    \smallskip
    
    \item $\vDash_{w}^{M, K} \Box A \Longleftrightarrow $ for all $w'$ with $w \leq_{K} w'$, if $ K \succ K'$ and $w' \in K'$ then $\vDash_{w'}^{M, K'} A$;

    \smallskip

    \item $\vDash_{w}^{M, K} \Diamond A \Longleftrightarrow $ there is a $K'$ with $ K \succ K'$ such that $w \in K'$ and $\vDash_{w}^{M, K'} A$;

    \smallskip
  
    \item For non-empty $\Gamma$, $\Gamma \vDash_{w}^{M, K} A \Longleftrightarrow \forall w' (w \leq_{K} w' ): \vDash_{w'}^{M, K} B$ for all $B \in \Gamma$ implies $ \vDash_{w'}^{M, K} A$;

\smallskip
    
     \item $\Gamma \vDash^{M, K} A \Longleftrightarrow \forall w (w \in W_{K}): \Gamma \vDash_{w}^{M, K} A$;

     \smallskip
     
     \item $\Gamma \vDash^{M} A \Longleftrightarrow \forall K (K \in F_{M}): \Gamma \vDash^{M, K} A$;

\smallskip

    \item $\Gamma \vDash A \Longleftrightarrow \Gamma \vDash^{M} A$ holds for all $M$.
\end{enumerate}

\end{definition}

\begin{definition}\label{def:validityinfullsimilaritymodel}
 Given any homogeneous model $H$, any $K \in F_{H}$ and any $w \in W_{K}$, the relations $(\vDash^{H, K}_{w})$,  $(\vDash^{H, K})$ $(\vDash^{H})$ and $(\vDash^{*})$ are defined as follows:
 \begin{enumerate}

    \item Clauses 1 through 5 and 8 through 10 are as in Definition \ref{def:validityinpartialsimilaritymodel}, but all occurrences of $K$ are replaced by $B$ and $\vDash$ is replaced by $\vDash^{*}$;

    \setcounter{enumi}{5}

\smallskip

    \item $\vDash_{w}^{H, K} \Box A \Longleftrightarrow $ for all $K'$, $K \succ K'$ implies $\vDash_{w}^{H, K'} A$;

\smallskip

  \item   $\vDash_{w}^{H, K} \Diamond A \Longleftrightarrow $ there is a $K'$ such that $ K \succ K'$ and $\vDash_{w}^{H, K'} A$;
    
\end{enumerate}

\end{definition}

Notice that $w \in W_{K}$ implies $w \in W_{K'}$ when both $K$ and $K'$ are in the same homogeneous model, so the respective conditions on the modal clauses for partial models can be dropped.  Moreover, since all propositional models in the same homogeneous model $H$ share the same structure, we can uniformly refer to their sets $W$ and relation $\leq$ without specification.

The definition of validity can be used to show monotonicity:

 \begin{lemma}\label{lemma:monotonicity}[Monotonicity]  The following results hold:

\begin{enumerate}

    \item  For all $w'$, $\Gamma \vDash^{M, K}_{w} A$ and $w \leq_{K} w'$ implies $\Gamma \vDash^{M, K}_{w'} A$;

\smallskip

    \item  For all $w'$, $\Gamma \vDash_{w}^{H, K} A$ and $w \leq w'$ implies $\Gamma \vDash^{H, K}_{w'} A$;
\end{enumerate}

\end{lemma}

%In the next two sections we show that partial models provide a semantics for $IK$ and homogeneous models provide a semantics for $MK$. The result is shown by presenting validity-preserving mappings to and from birelational models w.r.t. partial models, as well as to and from strong models w.r.t. homogeneous models.

In the next two sections we present validity-preserving mappings to and from birelational models w.r.t. partial models, as well as to and from strong models w.r.t. homogeneous models, thus showing that partial models provide a semantics for $IK$ and homogeneous models a semantics for $MK$.

\section{Mapping the new models into the traditional models}

We now define a mapping that transforms partial models into equivalent birelational models and homogeneous models into equivalent strong models. For the sake of convenience, we present the mapping itself for abstract models:

\begin{definition}  \label{def:definitionequivalencebirelationalmodel}

 For any abstract model $A = \langle F, \succ  \rangle$, its \textit{birelational equivalent} $A^{*} = \langle W^*,  \leq^{*}, R^{*}, v^{*} \rangle$ is defined as follows:

\begin{enumerate}
\item $W^* = F^{*}$, where $F^{*} = \{\langle w, K \rangle | w \in W_{K}$ and $K \in F\}$;

\item $v^{*} (\langle w, K \rangle) = v_{K}(w)$

    \item   $\langle w, K \rangle  \leq^{*} \langle w', K' \rangle $ iff $K = K'$ and $w \leq_{K} w'$;
    
    \item   $\langle w, K \rangle  R^{*} \langle w', K' \rangle $ iff $w = w'$ and $K \succ K'$;
\end{enumerate}
\end{definition}

So to every pair of a world $w$ with a model $K$ that contains it in $A$ we associate a world $\langle w, K \rangle $ in $A^{*}$, hence every distinct occurrence of the same world in the original model becomes a different world in the new model. The function $v^{*}$ assures us that an atom is valid in the $w$ of $K$ if and only if it is valid in $\langle w, K \rangle$, and the definition of $\leq^{*}$ guarantees that the ``internal structure" of each $K$ is preserved in the new model. Finally, the definition of $R^{*}$ makes it so that a world $\langle w, K \rangle$ can access a world $\langle w', K' \rangle$ only if $w$ and $w'$ are actually occurrences of the same worlds in any models $K$ and $K'$ such that $K \succ K'$, so the accessibility between models is transformed into an accessibility between worlds. Notice also that $\langle w, K \rangle R^{*} \langle w, K' \rangle$ only holds if both $K$ and $K'$ contain $w$, as $w \notin W_{K}$ implies $\langle w, K \rangle \notin F^{*}$ and $w \notin W_{K'}$ implies $\langle w, K \rangle \notin F^{*}$. In what follows we will be interested specifically in the birelational equivalents of partial models $M$ and homogeneous models $H$, denoted by $M^{*}$ and $H^{*}$.

This mapping is very similar to a mapping of birelational models into intuitionistic first-order models presented by Simpson in \cite[pg. 151]{Sim94}, the main difference being that Simpson uses pairs $\langle w, d\rangle$ of worlds and elements of their domain instead of pairs of models and elements of their frame

%The structure of this mapping is very similar to Simpson's mapping of birelational models into intuitionistic first-order logic, the main difference being that Simpson uses pairs $\langle w, d\rangle$ of worlds and elements of their domain instead of pairs of models and their worlds.

Now we prove the following:

\begin{lemma} \label{lemma:existencebirelationamodel}
 If $M$ is a partial model then $M^{*}$ is a birelational model.
\end{lemma}

    \begin{lemma} \label{lemma:correspondenceproofmappingsimilarityintobirelationa}
        $\Gamma \vDash^{M,K}_{w} A$ iff $\Gamma \Vdash^{M^{*}}_{\langle w, K\rangle} A$, for all $w \in W_{K}$ and $K \in F$.
    \end{lemma}

The results can also be from partial models to homogeneous models as follows:

\begin{lemma} \label{lemma:existencestrongbirelationamodel}
 If $H$ is a homogeneous model then $H^{*}$ is a strong model.
\end{lemma}

The result below is not used in any of our proofs, but it shows that the models obtained through the mapping satisfy condition $F_{4}$:

\begin{proposition}\label{prop:excessivemodels}
    If $\langle w, K \rangle \leq^{*} \langle w', K' \rangle$ and $\langle w'', K'' \rangle R^{*} \langle w', K' \rangle$ then there is a unique $\langle w''', K''' \rangle$ such that $\langle w''', K''' \rangle R^{*} \langle w, K \rangle$ and $\langle w''', K''' \rangle \leq^{*} \langle w'', K'' \rangle$.
\end{proposition}

 \begin{lemma} \label{lemma:correspondenceproofmappingshomogeneous}
        $\Gamma \vDash^{H,K}_{w} A$ iff $\Gamma \Vdash^{M^{*}}_{\langle w, K\rangle} A$, for all $w \in W_{K}$ and $H \in F$.
    \end{lemma}

\begin{proof}
   Proofs of all five results are in the Appendix.
\end{proof}

\section{Mapping the traditional models into the new models}

This direction is trickier because, in order to build one of the new models, we must fix a single frame and either use it as a reference frame (partial models) or as the only frame (homogeneous models). Thus, in order to provide a uniform mapping for all models, we must fix a frame in which every object $j$ has the greatest possible number of non-modally accessible objects $j \leq j'$. For the sake of simplicity, we restrict our attention only to frames with countably many objects, but our strategy can also be adapted to deal with uncountable sets. Just like before, we start by proving our results for partial models and then show how they can be adapted to homogeneous models.

\begin{definition}
 A \textit{canonical quasi-structure} is a pair $\langle Q, R\ \rangle$ such that:

\begin{enumerate}
    \item $Q_{n}$ and $Q_{m}$ are non-empty disjoint sets of worlds whenever $n, m \geq 0$ and $m \neq n$, and $Q = \bigcup_{n\geq 0}Q_{n} $ (we may specify the $Q_{n}$ to which a particular $j \in Q$ belongs when this is relevant);

    \item $R$ is a relation on the elements of $Q$ such that, for every $j \in Q_{n}$:

    \begin{enumerate}
        \item  There are infinitely many $j' \in Q_{n+1}$ such that $jRj'$;

        \item For all $n > 0$, there is exactly one $j' \in Q_{n - 1}$ such that $j'Rj$;

        \item For no other $j'$ it holds that $jRj'$ or  $j'Rj$.
       \end{enumerate}
    %there are infinitely many $j' \in Q_{n+1}$ such that $jRj'$, and for every $j' \in Q_{m} (m > 0)$ there is exactly one $j \in Q_{m - 1}$ such that $jRj'$.

    %and, for all $w'' \in Q$, if $w''Rw'$ then $w'' = w$.
\end{enumerate}

\end{definition}

So $R$ induces a tree structure in which every object $j\in Q_{n}$ has infinitely many ``dedicated accessible worlds" in the set $Q_{n +1}$, each coming with their own infinite set of dedicated worlds on $Q_{n+2}$ (and so on). Notice that it follows from the definition that $Q_{n}$ is infinite for every $n \geq 1$, but $Q_{0}$ is only required to be non-empty. As will become clear, this means that every birelational model can also be mapped specifically to a model in which the reference frame and each of its at least partial copies have unique ``roots".

\begin{definition}
    Let $\langle Q, R\rangle$ be a canonical quasi-structure. Then a \textit{full canonical structure} $FC$ is a pair $\langle Q_{FC}, \leq_{FC} \rangle$ where $Q = Q_{FC}$ and $\leq_{FC}$ is the reflexive and transitive closure of $R$.
\end{definition}

\begin{definition}\label{def:PCj}
    Let $PC$ be any at least partial copy of $FC$ and $j$ an arbitrary object with $j \in Q_{PC}$. Then $PC^{j} = \langle Q_{PC^{j}}, \leq_{PC^{j}} \rangle$, where $Q_{PC^j} = \{ j' | j \leq_{PC} j'\}$ and, for all $j', j'' \in Q_{PC^j}$, $j' \leq_{PC} j''$ iff $j' \leq_{PC^j} j''$.
\end{definition}

\begin{lemma}\label{lemma:atleastpartialcopyFC}
    If $PC$ is an at least partial copy of $FC$ then $PC^{j}$ is an at least partial copy of $FC$.
\end{lemma}

%\begin{remark}
%If $Q_{PC} = Q_{FC}$ then by Clause 3 clearly $\leq_{PC} = \leq_{FC}$ and so $FC = PC$, so a full canonical structure is also part of itself.
%\end{remark}

\begin{definition} \label{def:interpretationfunction}
    Let $B = \langle W, \leq, R, v \rangle$ be a birelational model and $w \in W$. Let $FC$ be any full canonical structure and $PC$ any at least partial copy of $FC$. A interpretation of $w$ in $PC$ is any function $f_{B, w}$ such that:

\begin{enumerate}
%\item For all $j \in Q_{0}$ with $j \in Q_{PC}$, $f_{B, w} (j) = v(w)$;

  %  

  \item There is some $j \in Q_{PC}$ such that $f_{B,w}(j) = v(w)$;

    \item For all $j \in Q_{PC}$, $f_{B, w}(j) = v(w')$ for some $w \leq w'$;

    \item For all $j \in Q_{PC}$, if $f_{B, w}(j) = v(w')$ and $w' \leq w''$, then there is a $j' \in Q_{PC}$ such that $j \leq_{PC} j'$ and $f_{B, w}(j') = v(w'')$;

    \item For all $j,j' \in Q_{PC}$, if $f_{B, w}(j) = v(w')$, $f_{B, w}(j') = v(w'')$ and $j \leq_{PC} j'$ then $w' \leq w''$.
    
    %\item For all $w' \in Q$, if $f_{B, w}(w') = v(w'')$ and $w'' \leq j$, then there is a $j' \in Q$ such that $w' \leq j'$ and $f_{B, w}(j') = v(j)$.
\end{enumerate}
    
\end{definition}

Notice that there must be at least one $j \in Q_{n}$ such that $f_{B, w} = v(w)$, where $n$ is the smallest index for which there is at least one $j'$ with $j' \in Q_{n}$ and $j' \in Q_{PC}$.  So an interpretation of $w \in W$ is a function $f_{B, w}$ from elements of $Q_{PC}$ to valuations $v(w')$ for $w' \in W$, with the constraint that if the function assigns to some $j \in Q_{PC}$ the set of atoms assigned to some $w' \in W$ then it must also assign to some $j \leq_{PC} j'$ the set of atoms assigned to $w''$ for every $w' \leq_{PC} w''$. Additionally, if it assigns to some $j' \in Q_{PC}$ the set of atoms assigned to some $w' \in W$ then it must assign to all $j\in Q_{PC}$ with $j \leq j'$ the set of atoms assigned to some $w'' \in W$ with $w'' \leq w'$. Intuitively, this means that if the atoms assigned to $w \in W$ are assigned to $j \in Q_{PC}$ then all the extensions of $w$ are assigned to some extension of $j$ (first condition), and nothing else is assigned to extensions of $j$ (second condition). As will be seen later, the purpose of a function $f_{B,w}$ is to make an entire propositional model behave as if it were the particular world $w$ of $B$, so the relation $\succ$ between frames also behaves as if it was a relation between worlds. 

%The definition of $PC^{j}$ then allow us to pick a frame behaving exactly like $w$ for any $w \in W$ which specifically has the object $j$ (for arbitrary $j$) as its ``root'', thus allowing us to prove some important existence lemmas.

First we show that such a function exists for every $PC$ and every $w$ occurring in a birelational model:

\begin{lemma} \label{lemma:firstfunctionexistence}
   Let $B = \langle W, \leq, R, v \rangle$ be a birelational model and $w$ some object $w \in W$. Let $FC$ be any full canonical structure and $PC$ an at least partial copy of $FC$. Then there exists a interpretation of $w$ in $PC$.
\end{lemma}

We also need to construct one final function which is essential to deal with modalities and that will finally allow us to bridge the gap between birelational models and the new models:

\begin{lemma} \label{lemma:modalexistence}

    Let $B = \langle W, \leq, R, v \rangle$ be a birelational model and $f_{B, w}$ a function interpreting $w \in W$ in a at least partial copy $PC$ of a full canonical structure $FC$. Let $w', w'' \in W$ be arbitrary objects with $w \leq w'$ and $w' R w''$. Let $j$ be any $j \in Q_{PC}$ with $f_{B,w}(j) = v(w')$. Then there is a function $f_{B, w''}$ interpreting $w''$ in $PC^j$ such that, for all $j' \in Q_{PC^{j}}$, if $f_{B, w}(j') = v(s)$ and $f_{B, w''}(j') = v (s')$ then $s R s'$.

\end{lemma}

Despite its lenght, we prove the result here (as opposed to in the Appendix) due to the novelty and importance of the proof, which is a new version of the traditional modal existence lemma:

\begin{proof}

Notice that since $Q_{PC^j} = \{j' | j \leq_{PC} j' \}$ then $j' \in Q_{PC^j}$ implies $j \leq_{PC} j'$, and since $f_{B,w}(j) = v(w')$ we may conclude by Clause 4 of Definition \ref{def:interpretationfunction} that if $f_{B,w}(j') = v(s)$ for some $j' \in PC^{j}$ then since $j \leq_{PC} j'$ we also have $w' \leq s$. With this in mind, we define the function $f_{B, w''}$ as follows:

%$j \leq_{PC}^* j'$ for any We define the function $f_{w'', B}$ as follows:

%First notice that if $f_{B,w} (j) = v(w''')$ then $w \leq w'''$ by Clause $1$ of the definition of $f_{B, w}$. With this in mind, we define the function $f_{w'', B}$ as follows:

\begin{enumerate}
    \item  $f_{B, w''}(j) = v(w'')$;

\item For all $j' \in Q_{PC^j}$ with $j' \neq j$, let $f_{B, w}(j') = v (s)$. Then $w' \leq s$, and, since $w' R w''$, from $F_{1}$ (cf. Definition \ref{F1andF2}) and witness uniqueness (cf. Clause 5 of Definition \ref{def:birelational}) we conclude that there must be a unique $s'$ such that $s R s'$ and $w'' \leq s'$. Then we define $f_{B,w''}(j') = v(s')$.

%From Clause 1 Definition \ref{def:interpretationfunction} we have $w \leq s$. Since $w \leq s$ and $w R$ Since $w \leq $ Then we define $f_{B, w''} (j') = v(s')$,  where $s'$ is the unique object such that $s R s'$ and $w'' \leq s'$;
\end{enumerate}

Since by definition $f_{B, w''}(j) = v(w'')$ the function clearly satisfies Clause 1 of Definition \ref{def:interpretationfunction}, and since by definition for all $j' \in Q_{PC^j}$ we have $f_{B, w''}(j') = v(s')$ for some $w'' \leq s'$ it clearly also satisfies Clause 2. Therefore, in order to show that $f_{B, w''}$ is indeed an interpretation of $w'' \in W$ in $PC^j$ it suffices to show that it satisfies the two remaining properties.

\begin{enumerate}

    \item For all $j' \in Q_{PC^j}$, if $f_{B, w''}(j') = v(s)$ and $s \leq s'$, then there is a $j'' \in Q_{PC^j}$ such that $j' \leq_{PC^j} j''$ and $f_{B, w''}(j'') = v(s')$.

    Assume $f_{B, w''}(j') = v(s)$ and $s \leq s'$ for some $j' \in Q_{PC^j}$. Since $j' \in Q_{PC^j}$ we conclude $j \leq_{PC} j'$ by Definition \ref{def:PCj}. From the definition of $f_{B, w''}$ we have $f_{B,w}(j') = v(s'')$ for some $s'' R s$. From $F_{2}$ and witness uniqueness we conclude that, since $s'' R s$ and $s \leq s'$, there must be some unique $s'''$ such that $s'' \leq s'''$ and $s''' R s'$. Since $s'' \leq s'''$ and $f_{B, w}(j') = v(s'')$ we conclude by Clause 2 of Definition \ref{def:interpretationfunction} that there is some $j'' \in Q_{PC}$ such that $j' \leq_{PC} j''$ and $f_{B,w}(j'') = v(s''')$. Since $j \leq_{PC} j'$ and $j' \leq_{PC} j''$ by transitivity of $\leq_{PC}$ we conclude $j \leq_{PC} j''$, so from Definition \ref{def:PCj} we get $j'' \in Q_{PC^j}$. Since $f_{B,w}(j'') = v(s''')$ then $f_{B,w''}(j'') = v(s'''')$, where $s''''$ is the unique object with $s''' R s''''$ and $w'' \leq s''''$. Since  $f_{B, w''}(j') = v(s)$ and $f_{B,w''}$ satisfies Clause 2 of Definition \ref{def:interpretationfunction} we conclude $w'' \leq s$, and since $s \leq s'$ by transitivity of $\leq$ we conclude $w'' \leq s'$. So we have $w'' \leq s'$, $s''' R s'$, $w'' \leq s''''$ and $s''' R s''''$, so from uniqueness of $s''''$ we conclude $s' = s''''$. Hence since  $f_{B,w''}(j'') = v(s'''')$ we have $f_{B,w''}(j'') = v(s')$, and since $j' \leq_{PC} j''$ we also have $j' \leq_{PC^j} j''$ by Definition \ref{def:PCj}, so $j''$ is the desired witness.

%Since $w'' \leq s'$, $s''' R s'$ and $s''''$ is unique we conclude $s' = s''''$, so $f_{B,w''}(j'') = v(s')$, and since $j' \leq_{PC} j''$ by Definition \ref{def:PCj} also $j' \leq_{PC^j} j''$, so $j''$ is the desired witness.

\item For all $j',j'' \in Q_{PC^j}$, if $f_{B, w''}(j') = v(s)$, $f_{B, w''}(j'') = v(s')$ and $j' \leq_{PC^j} j''$ then $s \leq s'$.

  Pick any two objects $j', j'' \in Q_{PC^j}$ with $f_{B, w''}(j') = v(s)$, $f_{B, w''}(j'') = v(s')$ and $j' \leq_{PC^j} j''$. By the definition of $f_{B,w''}$ we have $f_{B,w}(j') = v(s'')$ for  $s''$ with $s''Rs$ and $f_{B,w}(j'') = v(s''')$ for $s'''$ with $s'''Rs'$. Due to how $f_{B,w''}$ is defined we also have that $s'$ is the unique object such that $w'' \leq s'$ and $s''' R s'$. Since $j' \leq_{PC^j} j''$, by Definition \ref{def:PCj} we conclude $j' \leq_{PC} j''$, so since $f_{B,w}(j') = v(s'')$ and $f_{B,w}(j'') = v(s''')$ we conclude $s'' \leq s'''$ by Clause 4 of Definition \ref{def:interpretationfunction}. Since $s'' \leq s'''$ and $s''Rs$, by $F_{1}$ and witness uniqueness there must be a unique $s''''$ such that $s \leq s''''$ and $s'''Rs''''$. Since $f_{B, w''}(j') = v(s)$ and $f_{B,w''}$ satisfies Clause 2 of Definition \ref{def:interpretationfunction} we conclude $w'' \leq s$, so since $s \leq s''''$ we conclude $w'' \leq s''''$ by transitivity of $\leq$. Since $w'' \leq s'$, $s'''Rs'$, $w'' \leq s''''$ and $s'''Rs'''''$ and the object $s'$ is unique we conclude $s' = s''''$, hence since $s \leq s''''$ we conclude $s \leq s'$, as desired.

\end{enumerate}

\end{proof}

Notice that, since birelational models do not always satisfy $F_{3}$, use of limited frames such as $PC^j$ is absolutely essential for the proof to follow through. If we could not delete objects below $j$ (which happens when we consider homogeneous frames), consider what would happen if we had $f_{B,w}(j) = v(s)$ and $f_{B,w}(j') = v(s')$ for $j' \leq j$, as well as $f_{B,w''}(j) = v(s'')$ for some $sRs''$. Due to the definition of $f_{B,w}$ we have $s' \leq s$ Then, if we wanted to show that $f_{B,w}(j'') = v(w''')$ and $f_{B,w''}(j'') = v(w'''')$ implies $w'''Rw''''$ for all $j''$, we would have to find some $s'''$ such that $s'Rs'''$ in order to put $f_{B,w''}(j') = v(s''')$, and due to the conditions imposed on $f_{B,w''}$ we would also have $s''' \leq s''$. But this just means that whenever we have $s' \leq s$ and $sRs''$ we also require some $s'''$ with $s' R s'''$ and $s''' \leq s''$, which is precisely $F_{3}$. This problem is dealt with by frames such as $PC^j$ precisely because we remove all objects with $j' \leq j$ from the frame.

We are finally ready to define the partial model correspondent of a birelational model:

\begin{definition} \label{def:equivalentbirelationalmodel}
    Let $B = \langle W, \leq, R, v \rangle$ be a birelational model. Let $FC$ be any fixed full canonical structure. The partial model $M^B = \langle F, \succ \rangle$ is defined as follows:

    \begin{enumerate}
        \item $F = \{  \langle Q_{PC}, \leq_{PC}, f_{B, w} \rangle  | PC$ is an at least partial copy of $FC$ and $f_{B,w}$ is a interpretation of $w$ in $PC$ for some $w \in W\}$
        
        \item $\langle Q_{PC}, \leq_{PC},f_{B, w} \rangle \succ \langle Q_{PC'}, \leq_{PC'},f'_{B, w'} \rangle$ iff, for all $j \in Q_{PC'}$ with $j \in Q_{PC}$, if $f_{B, w}(j) = v(w'')$ and $f'_{B, w'}(j) = v(w''')$ then $w'' R w'''$.
        
        %for all $j \in Q_{PC}$, if $f_{B,w} (j) = v(s)$ and $j \in Q_{PC'}$ then 
    \end{enumerate}

\end{definition}

So our propositional models are given by all pairs comprised of a partial copy $PC$ of $FC$ and an interpretation function for a world of the original model in $PC$, and the accessibility relation $K \succ K'$ holds between two models of $F$ only when, if the function of $K$ assigns $s$ to some object in $W_{K}$ which is also in $W_{K'}$, then the function of $K'$ assigns some world $s'$ modally accessible from $s$ (that is, some $s'$ with $sRs'$) to the same object. Notice that whenever for some model $\langle Q_{PC}, \leq_{PC}, f_{B,w} \rangle$ we have $f_{B,w}(j) = v(w')$ and $w' R w''$, Lemma \ref{lemma:modalexistence} yields a model $\langle Q_{PC^j}, \leq_{PC^j}, f_{B,w'} \rangle$ such that $\langle Q_{PC}, \leq_{PC}, f_{B,w} \rangle \succ \langle Q_{PC^j}, \leq_{PC^j}, f_{B,w'} \rangle$, which is essential for dealing with the modal case of the inductive proofs that we will soon present.

%in case to every object in $W_{K'}$ that's also in $W_{K}$ the intepretation function assigns some world $s'$ that's modally accessible from the world $s$ assigned to the same object in $K$. Notice that whenever for some model $\langle Q_{PC}, \leq_{PC}, f_{B,w} \rangle$ we have $f_{B,w}(j) = v(w')$ and $w' R w''$ Lemma \ref{lemma:modalexistence} yields a model $\langle Q_{PC^j}, \leq_{PC^j}, f_{B,w'} \rangle$ such that $\langle Q_{PC}, \leq_{PC}, f_{B,w} \rangle \succ \langle Q_{PC^j}, \leq_{PC^j}, f_{B,w'} \rangle$, which is essential for dealing with the modal case of some inductive proofs soon to be presented.

\begin{corollary}\label{cor:MBisapartialmodel}
    $M^B$ is a partial model.
\end{corollary}

\begin{proof}
Since $FC$ is an at least partial copy of itself we can conclude that $\langle Q_{FC}, \leq_{FC}, f_{B,w} \rangle \in F$ for some $f_{B,w}$, so it is straightforward to check that $F$ is an at least partially homogeneous set having $\langle Q_{FC}, \leq_{FC}, f_{B,w} \rangle \in F$ as its reference model, so all conditions of Definitions \ref{def:partialcopy} and \ref{definition:atleastpartiallyhomogeneous} are satisfied.
    
\end{proof}

%\begin{definition}
 %   Given any interpretation function $f_{B,w}$,  the function $\alpha$ is said to be a \textit{extraction function} if $f_{B,w} (j) = v(w)$ implies $\alpha(f_{B,w} (j)) = w$.
%\end{definition}

%\begin{lemma}
%Let $\alpha$ be a extraction function. Then  $\vDash^{M^B, \langle Q_{PC}, \leq_{PC}, f_{B, w} \rangle}_{j} A$ iff $\Vdash^{B}_{\alpha(f_{B,w}(j))} A$.
%\end{lemma}

\begin{lemma} \label{lemma:proofofequalityforworlds}
 Let $B$ be an arbitrary birelational model. Let $K$ be an arbitrary $K \in F_{M^{B}}$ with $K = \langle Q_{PC}, \leq_{PC}, f_{B,w} \rangle$ and $f_{B,w}(j) = v(w')$ for some $j \in Q_{PC}$. Then $\Gamma \vDash_{j}^{M^B, K} A$ iff $\Gamma \Vdash^{B}_{w'} A$.
\end{lemma}

    \begin{corollary} \label{cor:firstcorollary}
    Let $B$ be an arbitrary birelational model and pick a $K \in F_{M^{B}}$ with $K = \langle Q_{PC}, \leq_{PC}, f_{B,w} \rangle$. Then $\Gamma \vDash^{M^B, K} A$ iff $\Gamma \Vdash^{B}_{w} A$.
    \end{corollary}

\begin{proof}
    Assume $\Gamma \vDash^{M^B, K} A$. Then for all $j \in Q_{PC}$ we have $\Gamma \vDash_{j}^{M^B, K} A$. By Clause 1 of Definition \ref{def:interpretationfunction} we have that there must be some $j' \in Q_{PC}$ such that $f_{B,w}(j') = v(w)$. Then by putting $j = j'$ we have $\Gamma \vDash_{j'}^{M^B, K} A$, and by Lemma \ref{lemma:proofofequalityforworlds} we conclude $\Gamma \Vdash^{B}_{w} A$. For the converse, assume $\Gamma \Vdash^{B}_{w} A$. Then by monotonicity we have $\Gamma \Vdash^{B}_{w'} A$ for all $w \leq w'$. Pick any $j \in Q_{PC}$ and, without loss of generality, let $f_{B,w}(j) = v(w')$. By Clause 2 of Definition \ref{def:interpretationfunction} we have $w \leq w'$, so $\Gamma \Vdash^{B}_{w'} A$. Then by Lemma \ref{lemma:proofofequalityforworlds} we have $\Gamma \vDash_{j}^{M^B, K} A$, and since $j$ was arbitrary we conclude $\Gamma \vDash^{M^B, K} A$.
\end{proof}

\begin{corollary} \label{cor:secondcorollary}
    $\Gamma \vDash^{M^B} A$ iff $\Gamma \Vdash^{B} A$.
\end{corollary}

\begin{proof}
    Assume $\Gamma \vDash^{M^B} A$. Then $\Gamma \vDash^{M^B, K} A$ for all $K \in F$. Pick any $w \in W_{B}$. From Lemma \ref{lemma:firstfunctionexistence} we conclude that there is a $K' \in F$ such that $K' = \langle Q_{PC}, \leq_{PC}, f_{B,w} \rangle$. Then by putting $K = K'$ we have $\Gamma \vDash^{M^B, K'} A$. Since $\Gamma \vDash^{M^B, K'} A$ and $K' = \langle Q_{PC}, \leq_{PC}, f_{B,w} \rangle$ it follows from Corollary \ref{cor:firstcorollary} that $\Gamma \Vdash^{B}_{w} A$. But $w$ is an arbitrary object with $w \in W_{B}$, so we conclude $\Gamma \Vdash^{B} A$. For the converse, assume $\Gamma \Vdash^{B} A$. So $\Gamma \Vdash_{w}^{B} A$ for every $w \in W$. Pick any $K \in F$. By Definition \ref{def:equivalentbirelationalmodel} we have $K = \langle Q_{PC}, \leq_{PC}, f_{B,w'} \rangle$ for some $w' \in W_{B}$. Since $\Gamma \Vdash_{w'}^{B} A$ then it follows from Corollary \ref{cor:firstcorollary} that $\Gamma \vDash^{M^B, K} A$. But $K$ was an arbitrary object with $K \in F$, so we conclude  $\Gamma \vDash^{M^B} A$.
\end{proof}

We now turn to the task of proving similar results for strong models with respect to homogeneous models. The full canonical structures $FC$ that were taken as reference frames of our partial models can now be taken as providing the fixed frame that we use to build homogeneous models. Notice that, since $FC$ is already an at least partial copy of itself and every strong model is also a birelational model, Definition \ref{def:interpretationfunction} and Lemma \ref{lemma:firstfunctionexistence} can still be used. For greater clarity, we denote functions $f_{B,w}$ defined for strong models $S$ by $f_{S,w}$. On the other hand, since it uses partial copies $PC^{j}$ of $FC$ which are no longer acceptable, Lemma \ref{lemma:modalexistence} cannot be used in the new mapping. We prove an adapted, stronger version of it instead:

\begin{lemma} \label{lemma:modalexistencehomogeneous}

   Let $S = \langle W, \leq, R, v \rangle$ be a strong model and $f_{S, w}$ a function interpreting $w \in W$ in a full canonical structure $FC$. Let $w', w'' \in W$ be arbitrary objects with $f_{S,w}(j) = v(w')$ for arbitrary $j \in Q_{FC}$ and $w'Rw''$. Then there is a function $f_{S, w'''}$ interpreting some $w'''$ in $FC$ such that $f_{S,w'''}(j) = v(w'')$ and, for all $j' \in Q_{FC}$, if $f_{S, w}(j') = v(s)$ and $f_{S, w'''}(j') = v (s')$ then $s R s'$.

%Let $S = \langle W, \leq, R, v \rangle$ be a strong model and $f_{B, w}$ a function interpreting $w \in W$ in a full canonical structure $FC$. Let $w', w'' \in W$ be arbitrary objects with $w \leq w'$ and $w' R w''$. Then there is a function $f_{B, w''}$ interpreting $w''$ in $FC$ such that, for all $j \in Q_{FC}$, if $f_{B, w}(j) = v(s)$ and $f_{B, w''}(j) = v (s')$ then $s R s'$.

%Let $j$ be any $j \in Q_{PC}$ with $f_{B,w}(j) = v(w')$. Then there is a function $f_{B, w''}$ interpreting $w''$ in $PC^j$ such that, for all $j' \in Q_{PC^{j}}$, if $f_{B, w}(j') = v(s)$ and $f_{B, w''}(j') = v (s')$ then $s R s'$.

\end{lemma}

As such, satisfaction of $F_{3}$ means that we can always pick some $w'''$ for which the desired function $f_{S,w'''}$ can be created even in the full frame. Then we proceed to prove the analogues of Lemma \ref{lemma:proofofequalityforworlds} and Corollaries \ref{cor:firstcorollary} and \ref{cor:secondcorollary} for homogeneous models and strong models.

\begin{definition} \label{def:equivalenstronghomogeneousmodel}
    Let $S = \langle W, \leq, R, v \rangle$ be a strong model. Let $FC$ be any fixed full canonical structure. The homogeneous model $M^S = \langle F, \succ \rangle$ is defined as follows:

    \begin{enumerate}
        \item $F = \{  \langle Q_{FC}, \leq_{FC}, f_{S, w} \rangle  |$ $f_{S,w}$ is a interpretation of $w$ in $FC$ for some $w \in W\}$

        \item $\langle Q_{FC}, \leq_{FC},f_{S, w} \rangle \succ \langle Q_{FC'}, \leq_{FC'},f'_{S, w'} \rangle$ iff, for all $j \in Q_{FC}$, if $f_{S, w}(j) = v(w'')$ and $f'_{S, w'}(j) = v(w''')$ then $w'' R w'''$.
        
        %for all $j \in Q_{PC}$, if $f_{B,w} (j) = v(s)$ and $j \in Q_{PC'}$ then 
    \end{enumerate}

\end{definition}

\begin{corollary}\label{cor:MSishomogeneous}
   $M^S$ is a homogeneous model.
\end{corollary}

\begin{proof}
The result follows directly from the fact that all propositional models use the frame $\langle Q_{FC}, \leq_{FC} \rangle$, thus satisfying Definitions \ref{def:homogeneouset} and \ref{def:homogeneousmodel}.
    
\end{proof}

\begin{lemma} \label{lemma:proofofequalityforworldshomogeneous}
 Let $S$ be an arbitrary strong model. Let $K$ be an arbitrary $K \in F_{M^{S}}$ with $K = \langle Q_{FC}, \leq_{FC}, f_{B,w} \rangle$ and $f_{S,w}(j) = v(w')$ for some $j \in Q_{FC}$. Then $\Gamma \vDash_{j}^{M^S, K} A$ iff $\Gamma \Vdash^{S}_{w'} A$.
\end{lemma}

    \begin{corollary} \label{cor:firstcorollaryhomogeneous}
    Let $S$ be an arbitrary strong model and pick a $K \in F_{M^{S}}$ with $K = \langle Q_{FC}, \leq_{FC}, f_{S,w} \rangle$. Then $\Gamma \vDash^{M^S, K} A$ iff $\Gamma \Vdash^{S}_{w} A$.
    \end{corollary}

\begin{proof} The proof is analogous to that of Corollary \ref{cor:firstcorollary}, the main difference being that we use Lemma \ref{lemma:proofofequalityforworldshomogeneous} instead of Lemma \ref{lemma:proofofequalityforworlds}.

\end{proof}

\begin{corollary} \label{cor:secondcorollaryhomogeneous}
    $\Gamma \vDash^{M^S} A$ iff $\Gamma \Vdash^{S} A$.
\end{corollary}

\begin{proof}

The proof is analogous to that of Corollary \ref{cor:secondcorollary}, the main difference being that we use Corollary \ref{cor:firstcorollaryhomogeneous} instead of Corollary \ref{cor:firstcorollary}.
\end{proof}

%The proof of Lemma \ref{lemma:modalexistencehomogeneous} does not differ significantly from the proof of Lemma \ref{lemma:modalexistence}. The main difference is that, in the proof of Lemma \ref{lemma:modalexistence}, we take a particular $j \in Q_{PC}$, create a limited frame $PC^j$ such that $j' \in Q_{{PC^j}}$ implies $j \leq j'$ and define the function $f_{B, w''}$ for it. This is essential for the proof of the modal cases in Lemma \ref{lemma:proofofequalityforworlds} to work; if instead we tried to create the function $f_{B},w''$ for $FC$, then we would also have to assign interpretations to all objects $j'' \leq j$ (which then could not be deleted) in such a way that if $f_{B,w}(j'') = s$ and $f_{B,w''}(j'') = s'$ then $s \succ s'$. Now consider what happens in the case of $\Box$. Pick $j$ and $j'$ with $j \leq j'$ and put $f_{B,w}(j) = w$, $f_{B,w}(j') = w'$, $f_{B,w''}(j') = w''$, as well as $w \leq w'$ and $w' R w''$. Since birelational models do not always satisfy $F_{3}$, there is no guarantee that there will be some $w''''$ with $wRw''''$ and $w'''' \leq w''$ for we to put $f_{B,w''}(j) = w''''$ in order to guarantee that $f_{B,w''''}$ has the desired property.

We are finally ready to prove our main results for $IK$ and $MK$:

\begin{theorem}[Soundness and Completeness]  The following hold:

\begin{enumerate}
    \item $\Gamma \vDash A$ iff $\Gamma \Vdash A$;

    \item $\Gamma \vDash^{*} A$ iff $\Gamma^{*} \Vdash A$.
\end{enumerate}
\end{theorem}

\begin{proof}
    For the first statement, assume  $\Gamma \vDash A$. Then $\Gamma \vDash^{M} A$ for all $M$. Let $B$ be any birelational model. Then by Definition \ref{def:equivalentbirelationalmodel} and Corollaries \ref{cor:MBisapartialmodel} and \ref{cor:secondcorollary} there is a partial model $M^{B}$ such that $\Gamma \vDash^{M^B} A$ iff $\Gamma \Vdash^{B} A$. By putting $M = M^{B}$ we have $\Gamma \vDash^{M^{B}} A$, so we conclude $\Gamma \Vdash^{{B}} A$, and since $B$ was arbitrary we conclude $\Gamma \Vdash A$. For the converse, assume $\Gamma \Vdash A$. Then $\Gamma \Vdash^{B} A$ for all $B$. Let $M$ be any partial model. Then by Definition \ref{def:definitionequivalencebirelationalmodel}, Lemma \ref{lemma:existencebirelationamodel} and Lemma \ref{lemma:correspondenceproofmappingsimilarityintobirelationa} there is a birelational model $M^{*}$ such that $\Gamma \vDash^{M} A$ iff $\Gamma \Vdash^{M^*} A$.  By putting $M^{*} = B$ we have $\Gamma \Vdash^{M^{*}} A$, so we conclude $\Gamma \vDash^{{M}} A$, and since $M$ was arbitrary we conclude $\Gamma \vDash A$. The proof of the second statement is analogous, the main difference being that it uses Definition \ref{def:equivalenstronghomogeneousmodel} and Corollaries \ref{cor:MSishomogeneous} and \ref{cor:secondcorollaryhomogeneous} instead of Definition \ref{def:equivalentbirelationalmodel} and Corollaries \ref{cor:MSishomogeneous} and \ref{cor:secondcorollary}, as well as Lemmas \ref{lemma:existencestrongbirelationamodel} and \ref{lemma:correspondenceproofmappingshomogeneous} instead of Lemmas \ref{lemma:existencebirelationamodel} and Lemma \ref{lemma:correspondenceproofmappingsimilarityintobirelationa}.

\end{proof}

Consider also Definition \ref{excessivemodels} and Proposition \ref{prop:excessivemodels}. If the definition is combined with the clauses from Definition \ref{def:validityinstrong}, the mapping from strong models into homogeneous models (and vice-versa) can still be used, meaning that all of the results of this section still hold. This shows that, as previously claimed, $MK$ is indeed maximal with respect to the listed frame conditions. On a final note, since the full canonical structure can be rooted (that is, its set $Q_{0}$ can be a singleton) and the partial structures $PC^j$ we use during the proofs for $IK$ are also rooted, the results of this section and of the one preceding it could also be shown if we considered only rooted partial and homogeneous models.

\section{Extension to other non-classical logics} 

The logic $MK$ holds a distinguished position in this framework due to the fact that it allows a straightforward generalization of the modal semantics to other logics. To see why, consider Definition \ref{def:validityinpartialsimilaritymodel} and the shape of the clause for $\Box$ in $IK$:

%The logic $MK$ holds a distinguished position in this framework due to the fact that it allows a straightforward generalization of the modal semantics to other logics. To see why, consider Definition \ref{def:validityinpartialsimilaritymodel} and the shape of the clause for $\Box$ in $IK$:

\begin{itemize}
    \item[] $\vDash_{w}^{M, K} \Box A \Longleftrightarrow $ for all $w'$ with $w\leq_{K'} w'$, if $ K \succ K'$ and $w' \in K'$ then $\vDash_{w'}^{M, K'} A$;
\end{itemize}

Generalization of this clause is precluded by the fact that it relies on the intuitionistic modal relation $\leq$, without which monotonicity is lost. On the other hand, consider the modal clauses for $MK$ in Definition  \ref{def:validityinfullsimilaritymodel}:

\begin{enumerate}
  \item[] $\vDash_{w}^{H, K} \Box A \Longleftrightarrow$ for all $K'$, $K \succ K'$ implies $\vDash_{w}^{H, K'} A$;

    \smallskip

  \item[]   $\vDash_{w}^{H, K} \Diamond A \Longleftrightarrow $ there is a $K'$ such that $ K \succ K'$ and $\vDash_{w}^{H, K'} A$;
\end{enumerate}

There is nothing inherently intuitionistic about those clauses, so they only yield semantics for $MK$ because the propositional models $K$ themselves are intuitionistic. As such, we can obtain the $MK$ analogue of any logic characterizable through Kripke models by using these clauses and only taking into account homogeneous models of that logic.

The maximality of $MK$ with respect to the listed frame conditions suggest that every modal logic obtained through this method will be quite strong. In case the chosen logic does not satisfy modal categoricity, it may be possible to obtain distinct modal extensions by changing either the semantic clauses, the conditions imposed on frames or the class of models, but since the required changes depend on the particular logic being considered there can be no modular approach to such extensions. Just like $IK$ is obtained from $MK$ by allowing partially copied frames but changing the clause for $\Box$ in order to retain monotonicity, the shape of the modal clauses and of the propositional models will be dictated by which properties we want to induce in the modal logics.

In this paper we have only dealt with intuitionistic analogues of the classical modal logic $K$, but this is only so that the results are as general as possible. It can be easily shown that by imposing the traditional conditions on the modal relation $\succ$ we can obtain the analogues of stronger classical modal logics. Although in the case of $IK$ the proof of this fact for formulas containing $\Box$ relies on reasoning with the intuitionistic relation $\leq$, this can be shown to hold in general for all the analogues of $MK$ obtained through the modular modal clauses. We exemplify by considering two cases:

\begin{proposition}

 Given the modal clauses in Definition \ref{def:validityinfullsimilaritymodel} and any $K$, the following hold:

\begin{enumerate}
    \item  If $\vDash^{H, K}_{w} \Box A$ and $\succ$ is reflexive then $\vDash^{H, K}_{w} A$.

\smallskip
    
     \item  If $\vDash^{H, K}_{w} \Box A$ and $\succ $ is transitive then $\vDash^{H, K}_{w} \Box \Box A$.
\end{enumerate}
\end{proposition}

\begin{proof}
    For the first statement, assume $\vDash^{H, K}_{w} \Box A$. By reflexivity of $K$ we conclude $K \succ K$, so by the semantic clause for $\Box$ we obtain $\vDash^{H, K}_{w} A$. For the second, assume $\vDash^{H, K}_{w} \Box A$. In order to show a contradiction, further assume $\nvDash^{H, K}_{w} \Box \Box A$. So there must be a $K'$ with $K \succ K'$ such that $\nvDash^{H, K'}_{w} \Box A$, hence there must also be a $K''$ with $K' \succ K''$ such that $\nvDash^{H, K''}_{w} A$. Since $K \succ K'$, $K' \succ K''$ and $\succ$ is transitive we have $K \succ K''$, so the semantic clause for $\Box$ yields $\vDash^{M, K''}_{w} A$. Contradiction. Therefore, $\vDash^{H, K}_{w} \Box \Box A$.
\end{proof}

The proof of this proposition only relies on properties of $\succ$ and on the shape of our modal clauses, meaning that this holds for any non-classical modal logic obtained through this approach.

On a final note, notice that Tarskian truth-functional models can be visualized as degenerate Kripke models (usually models $K$ in which $W_{K}$ is a singleton), whence it follows that our approach also applies to logics which Tarskian semantics. In particular, intuitionistic propositional models with a singleton $W_{K}$ can easily be shown to correspond to classical propositional models, and it is equally easy to show that the logic obtained by applying this method to such models is precisely the classical modal logic $K$.

%There is also a very strong sense in which the approach we have presented here can be interpreted as a generalization of the traditional Kripkean approach to modal logics. Since Tarskian truth-functional models can be visualized as degenerate Kripke models (usually a model $K$ in which $W_{K}$ is a singleton), it follows that this approach also applies to logic with Tarskian models. In particular, intuitionistic propositional models with a singleton $W_{K}$ can easily be shown to correspond to classical propositional models, and it is equally easy to show that the logic obtained by applying this method to such models is precisely the classical modal logic $K$. In other words, the traditional approach can be characterized as a \textit{particular case} of the new approach, in which we only consider degenerate Kripke models.

\section{Higher-order Kripke models}  The approach to modal logics presented here generalizes the Kripkean approach by applying its ideas in a broader setting. We still start with a set of models, define an accessibility relation and use the resulting structure to define modalities, but now we are allowed to start with Kripke models instead of being limited to truth-assignment functions. Since our relational structures now contain relational structures themselves, this also leads to a \textit{mathematical} generalization of Kripke models, as the freshly obtained relational structure can then be used as an object in yet another relational structure. Following this idea to its natural conclusions leads us to \textit{higher-order Kripke models}, a generalization of Kripke models which is closer to Kripke's original idea (both conceptually and mathematically) than either neighborhood semantics \cite{Pacuit2017-PACNSF} or general (algebraic) Kripke models \cite{Chagrov1997-CHAML-2}, its other main generalizations.

Inasmuch they use more than one relation, birelational models themselves are generalizations of Kripke models. We can generalize this further by considering an arbitrary \textit{set} $\mathcal{R}$ of relations. After this is done, (multirelational) higher-order Kripke models can be defined as follows:

\begin{definition}
 Let $\mathbb{L}$ be a language and $\mathcal{V}$ a set of truth values. A \textit{$0$-ary Kripke model $K^{0}$} is a sequence $K^{0} = \langle W, \mathcal{R}, \textit{f} \rangle$, where $W$ is a non-empty set of objects, $\mathcal{R}$ a non-empty set of relations on $W$ and \textit{f} a function $f : W \times \mathbb{L} \mapsto \mathcal{V}$.
\end{definition}

\begin{definition}
 Let $\mathbb{L}$ be a language and $\mathcal{V}$ a set of truth values. For $n > 0$, a \textit{$n$-ary Kripke model $K^{n}$} is a sequence $K^{n} = \langle W, R, \textit{f} \rangle$, where $W$ is a non-empty set of $(n-1)$-ary Kripke models, $R$ a non-empty set of relations on $W$ and \textit{f} is a function $f : W \times \mathbb{L} \mapsto \mathcal{V}$.
\end{definition}

If we want to remain even closer to the original formulation, we can also define unirelational higher-order Kripke models as follows:

\begin{definition}
    A $0$-ary Kripke model is \textit{unirelational} if its set $\mathcal{R}$ is a singleton.
\end{definition}

\begin{definition}
    A $n$-ary Kripke models is \textit{unirelational} if its set $\mathcal{R}$ is a singleton and all elements of $W$ are unirelational $(n-1)$-ary Kripke models.
\end{definition}

Traditional Kripke models are then equivalent to $0$-ary unirelational Kripke models, and birelational models to $0$-ary Kripke models in which $\mathcal{R}$ has two elements. The function $f$ could alternatively been defined as assigning truth values only to pairs comprised of an element of $W$ and an atomic formula $\mathbb{L}$ if we also defined semantic clauses extending the valuations from atoms to the remaining formulas in $\mathbb{L}$, similar to what is done in Definitions \ref{def:validityinpartialsimilaritymodel} and \ref{def:validityinfullsimilaritymodel}. We could also loosen the definition of models by stating that the set of objects of a $n$-ary model must be models of level \textit{at most} $(n-1)$, but it is not clear whether this brings any benefit.

%Instead of using a function that assignsemantic clauses like those of Definitions and \ref{def:validityinpartialsimilaritymodel} and \ref{def:validityinfullsimilaritymodel} in order to extend 

%We could also have defined a more general notion in which a $n$-ary model is a model whose $W$ contains Kripke models of level \textit{at most} $(n-1)$, but it is not clear whether this further generalization brings any benefits. 

%A couple of other generalizations would also be possible, altough it is not immediatelly clear if they were bring any benefits.

Although this will not be investigated here, there are reasons to believe that the switch from $0$-ary to $n$-ary Kripke models leads to increased expressivity, in the sense that some Kripke-incomplete logics (thus not complete for $0$-ary models) might be complete with respect to some class of $n$-ary models for $n > 0$. Every $n$-ary Kripke model with $n > 0$ has at their disposal the functions used to define the $n-1$ ary models, meaning that validity in a $n$-ary model can be defined as a function of validity (or invalidity) in any given collection of $(n-1)$-ary models. For instance, we can define $1$-ary models $K^{1} = \langle W^{1}, \mathcal{R}^{1}, \textit{f}^{1}  \rangle$ for a logic $\mathcal{L}^{1}$ as containing only $0$-ary models $K = \langle W_{K}, \mathcal{R}_{K}, \textit{f}_{K}  \rangle$ sound and complete w.r.t. any other Kripke-complete logic $\mathcal{L}^{0}$ and then define its function $f^{1}$ as assigning the value true to a pair $\langle K, A \rangle$ if either $f_{K}$ assigns true to $\langle w, A \rangle$ for every $w \in W_{K}$ or some additional condition is satisfied, a procedure that guarantees that every theorem of $\mathcal{L}^{0}$ is still a theorem of $\mathcal{L}^{1}$ but may also allow the inclusion of new theorems through the additional condition of $f^{1}$.

Precisely stated, the conjecture is as follows:

\begin{conjecture} \label{conj:atleastone}
    There is at least one Kripke-incomplete logic $\mathcal{L}$ which is complete w.r.t. to a class of $n$-ary Kripke models for some $n > 0$.
\end{conjecture}

Such a result would be remarkable because, as mentioned before, higher-order Kripke models are much closer to the original idea behind Kripke models than neighborhood semantics and general Kripke models. Those generalizations often have to be used because, unlike traditional Kripke semantics, neighborhood semantics and general models can provide semantics for any normal modal logic. Due to its additional structure, we conjecture that this is also the case for the new semantics:

\begin{conjecture} \label{conj:allnormalmodal}
    All normal modal logics are complete w.r.t. some class of $n$-ary models for $n \geq 0$.
\end{conjecture}

Of course, this can only be the case if the innermost Kripke models of a given $n$-ary model (that is, the set $W$ of $(n-1)$-ary models) are used in an essential way by the definitions. This is not the case in our definitions for intuitionistic logic. In fact, notice that the specification of abstract models in Definition \ref{def:generalmodel} does not include a function at all, which does not bring about issues because Clause 9 of Definition \ref{def:validityinpartialsimilaritymodel} is essentially defining a function which assigns a designated truth value to $\langle K, A \rangle$ whenever the designated value is also assigned to $\langle w, A \rangle$ for every $w \in W_{K}$. Definitions in which the truth assignments of each $n$-ary model are fully reducible to the truth assignments at their $(n-1)$-ary objects cannot reasonably be expected to be more expressive than multirelational Kripke models because they ultimately reduce to valuations of $0$-ary models, meaning that in all likelihood it will be possible to provide a mapping transforming such $n$-ary models into multirelational $0$-ary models. On the other hand, our proofs for $IK$ and $MK$ essentially show that their $0$-ary birelational models can be transformed into $1$-ary unirelational models by pushing one of the relations upwards and making the function of the $1$-ary model agree with the functions of its $0$-ary models. This is likely a general phenomenon for finite sets $\mathcal{R}$, which leads us to some definitions plus a conjecture:

\begin{definition}
    A $0$-ary Kripke model is \textit{finitely relational} if its set $\mathcal{R}$ is finite.
\end{definition}

\begin{definition}
    A $n$-ary Kripke models is \textit{finitely relational} if its set $\mathcal{R}$ is finite and all elements of $W$ are finitely relational $(n-1)$-ary Kripke models.
\end{definition}

\begin{conjecture}\label{conj:pushup}
     If a logic is complete w.r.t. a class of $n$-ary finitely relational models then it is complete w.r.t. a class of unirelational $m$-ary models for $m \geq n$.
\end{conjecture}

\begin{conjecture} \label{conj:atleastonefinitary}
    There is at least one Kripke-incomplete logic $\mathcal{L}$ which is complete w.r.t. to a class of unirelational $n$-ary Kripke models for some $n > 0$.
\end{conjecture}

Meaning that it is likely the case that we cannot always push a relation or set of relations down in the model (which would make any $n$-ary model equivalent to some multirelational $0$-ary model) because it would also be necessary to push the valuation function down, and it seems that this can only be done in cases where the $n$-ary function could already have been expressed solely in terms of the $(n-1)$-ary functions (as in the intuitionistic case). On the other hand, it is likely that we can always push a relation or set of relations \textit{upwards}, precisely because if we start with $n$-ary models we can define $(n+1)$-ary models whose functions could have been expressed solely in terms of the $n$-ary models and use them just to put the relations at different levels. Since finitely relational models already might be more expressive than traditional models, this leads to the following conjecture:

%which would be a corollary of Conjectures \ref{conj:allnormalmodal} and \ref{conj:pushup}:

\begin{conjecture} \label{ref:unirelationalallmodal}
     All normal modal logics are complete w.r.t. some class of unirelational $n$-ary models for $n \geq 0$.
\end{conjecture}

Which would be a refinement of Conjecture \ref{conj:allnormalmodal}, as it would show that even it the version closest to Kripke's original formulation the new models would be sufficient for providing semantics for normal modal logics.

Higher-order Kripke models can be further generalized to infinitely layered structures in which we pick models for all $n \in \mathbb{N}$ instead of stopping at some $n$:

\begin{definition}
     A \textit{infinitary higher-order Kripke model} is a infinite sequence $\langle \mathcal{K}^{0}, \mathcal{K}^{1}, ... \rangle$ such that $\mathcal{K}^{0}$ is a set of $0$-ary Kripke models and, for all $n > 0$, $\mathcal{K}^{n}$ is a set of $n$-ary Kripke models such that $K \in \mathcal{K}^{n}$ implies $W_{K} \subseteq \mathcal{K}^{(n-1)}$.
\end{definition}

This concept of model could be paired, for instance, with definitions stating that a formula $A$ is valid in a model $K^{\infty}$ iff $\langle K, A \rangle$ receives the designated value for every $K$ which is an element of a set in occuring in $K^{\infty}$, or even if all models in the sequence satisfy some other condition. We can still expect Conjectures \ref{conj:atleastone}, \ref{conj:allnormalmodal}, \ref{conj:atleastonefinitary} and \ref{ref:unirelationalallmodal} to hold for the same reasons they are expected to hold for non-infinitary models, and the following infinitary variation of Conjecture \ref{conj:pushup} can also be stated:

\begin{conjecture}\label{conj:unirelationalinfinitary}
     If a logic is complete w.r.t. a class of infinitary models in which only finitely relational models occur then it is complete w.r.t. a class of infinitary models in which only unirelational models occur.
\end{conjecture}

This is expected to hold because if there are $m$ relations in a model $\mathcal{K}^{0}$ that need to be pushed up we could push one to the level $\mathcal{K}^{1}$, one to the level $\mathcal{K}^{2}$ and so on until the last one is put in $\mathcal{K}^{m}$. Then, we could put one of the relations originally occurring at the level $\mathcal{K}^{1}$ in the level $\mathcal{K}^{(m + 1)}$, one in $\mathcal{K}^{(m + 2)}$ and so on, eventually putting the last one in the level $\mathcal{K}^{(m + j)}$ where $j$ is the number of relations originally at the level $\mathcal{K}^{1}$. It seems that this could be repeated for all $n$, yielding a model with the desired property.

It would not be entirely unreasonable to expect that infinitary models, especially ones containing infinitary sets of relations, could be more expressive than neighborhood semantics and general Kripke models, which are not capable of providing semantics for all propositional modal logics \cite{GersonInadequacy51cfe23e-c1e6-3725-ade2-e8600362b7a3}\cite{Dahn1976-DAHNSA}. We thus finish the section by stating the boldest of our conjectures:

\begin{conjecture}
    Every propositional modal logic is complete with respect to some class of infinitary higher-order Kripke models.
\end{conjecture}

Which would mean that the new generalization fixes the main model-theoretic issue observed in standard generalizations of Kripke semantics.

\section{Conclusion} The mathematical and conceptual elegance of Kripke semantics led many to believe it would be the definitive semantics for normal modal logics, which turned out not to be the case due to the existence of many Kripke-incomplete logics \cite{KTBeab764a2-304d-3126-97ce-0631fd61611b}. In this paper we have shown that an equally elegant modal semantics (which is even closer to Kripke's original idea than neighborhood semantics and general Kripke semantics) can be provided for all non-classical modal logics whose propositional fragment is sound and complete with respect to some class of Kripke models. This is done by presenting such semantics for the intuitionistic modal logics $IK$ and $MK$, which are interesting in their own right. Higher-order Kripke models (both in their finitary and infinitary versions) are then obtained as a natural generalization of the semantics for $IK$ and $MK$, and the structure of some mappings defined between the new and the traditional (birelational) models for those logics hint at some general properties that can be reasonably expected to hold for all higher-order Kripke models.

Since this paper presents both results for modal logics in general and intuitionistic modal logics in particular, it will be interesting to further develop our investigations in both directions in future works. We are currently working on showing that, if semantics clauses with built-in monotonicity for both $\Box$ and $\Diamond$ are used and all propositional models are required to be intuitionistic, then abstract models (cf. Definition \ref{def:generalmodel}) are sound and complete with respect to the intuitionistic modal logic $WK$ \cite{Wijesekera1990-WIJCML}, as well as that a semantics for the logic
$CK$ \cite{bellin2001extended}\cite{degroot2025semanticalanalysisintuitionisticmodal} can be obtained through some additional adaptations. This would connect nicely with other studies focused on $CK$ and $WK$, such as the one in \cite{degroot2025semanticalanalysisintuitionisticmodal}. As for the other direction, a natural starting point would be to further investigate our conjectures for higher-order Kripke models, especially the ones about the expressive power of the framework. If the framework is found to be more expressive than neighbourhood or algebraic semantics (which is more likely in the case of models with infinitely many layers), this would be a major breakthrough in the study of semantics for non-normal modal logics.

\appendix

\section{Formal proofs}

The proofs are itemized, presented by section, and organised in order of appearance.

\medskip

\subsection{Proofs of Section 3}

\begin{itemize}

\item \textbf{Proposition \ref{prop:differenceIKMKfirst}}

\begin{proof}

For the first statement, let $S$ be a strong model with objects $W$. Pick any $w \in W$ and any $w \leq w'$ with $\Vdash^{S}_{w'} \neg \Box \bot$. Since $\Vdash^{S}_{w'} \neg \Box \bot$ we conclude that $w'\leq w''$ implies $\nVdash^{S}_{w''}  \Box \bot$, and since $w' \leq w'$ we have $\nVdash^{S}_{w'}  \Box \bot$, hence there must be a $j$ with $w' R j$ and $\nVdash^{S}_{j} \bot$. Since $\top$ is a theorem of $MK$ we have $\Vdash^{*} \top$ and thus $\Vdash^{S} \top$ and $\Vdash^{S}_{j} \top$, so since $w' R j$ we conclude $\Vdash^{S}_{w'} \Diamond\top$. Since we picked arbitrary $w'$ with $w \leq w'$ we conclude $\Vdash^{S}_{w} (\Box \bot ) \to (\Diamond \top)$, hence from arbitrariness of $w$ and $S$ we conclude $\Vdash^* (\Box \bot ) \to (\Diamond \top)$.

%For the first statement, let $S$ be a strong model with objects $W$. Pick any $w \in W$ and any $w \leq w'$ with $\Vdash^{S}_{w'} \neg \Box \bot$. Since $\Vdash^{S}_{w'} \neg \Box \bot$ by Clauses 4 and 5 of Definition \ref{def:validityinstrong} we conclude that $w'\leq w''$ implies $\nVdash^{S}_{w''}  \Box \bot$, so since $w' \leq w'$ we have $\nVdash^{S}_{w'}  \Box \bot$. By Clause 9 there must be a $j$ such that $w' R j$ and $\nVdash^{S}_{j} \bot$. Since $\top$ is a theorem of $MK$ we have $\Vdash^{*} \top$ and thus $\Vdash^{S}_{j} \top$ by Clauses 7 and 8, so since $w' R j$ we conclude $\Vdash^{S}_{w'} \Diamond\top$ by Clause 10. Since we picked arbitrary $w'$ with $w \leq w'$ we conclude $\Vdash^{S}_{w} (\Box \bot ) \to (\Diamond \top)$ by Clauses 4 and 6, hence from arbitrariness of $w$ and $S$ we conclude $\Vdash^* (\Box \bot ) \to (\Diamond \top)$ by Clauses 7 and 8.

For the second statement, let $S$ be a strong model with objects $W$. Pick any $w \in W$ and any $w \leq w'$ with $\Vdash^{S}_{w'} (\Box (A \lor \neg A) \land \neg \Box A)$. Then both $\Vdash^{S}_{w'} \Box (A \lor \neg A)$ and $\Vdash^{S}_{w'} \neg \Box A$ hold. Since $\Vdash^{S}_{w'} \neg \Box A$ we conclude $\nVdash^{S}_{w'} \Box A$, so there must be a $j$ such that $w' R j$ and $\nVdash^{S}_{j} A$. Since $\Vdash^{S}_{w'} \Box (A \lor \neg A)$ and $w'Rj$ we have $\Vdash^{S}_{j} A \lor \neg A$, thus either $\Vdash^{S}_{j} A$ or $\Vdash^{S}_{j} \neg A$, hence since $\nVdash^{S}_{j} A$ we conclude $\Vdash^{S}_{j} \neg A$. Since $w'Rj$ and $\Vdash^{S}_{j} \neg A$ we conclude $\Vdash^{S}_{w'} \Diamond \neg A$, hence by arbitrariness of $w'$ we conclude $\Vdash^{S}_{w} (\Box( A \lor \neg A) \land \neg \Box A) \to (\Diamond \neg A)$ and by arbitrariness of $w$ and $S$ we conclude $\Vdash^{*} (\Box( A \lor \neg A) \land \neg \Box A) \to (\Diamond \neg A)$.

%For the second statement, let $S$ be a strong model with objects $W$. Pick any $w \in W$ and any $w \leq w'$ with $\Vdash^{S}_{w'} (\Box (A \lor \neg A) \land \neg \Box A)$. Then by Clause 2 both $\Vdash^{S}_{w'} \Box (A \lor \neg A)$ and $\Vdash^{S}_{w'} \neg \Box A$ hold. Since $\Vdash^{S}_{w'} \neg \Box A$ we conclude $\nVdash^{S}_{w'} \Box A$ by Clauses 4 and 5, so by Clause 9 there must be a $j$ such that $w' R j$ and $\nVdash^{S}_{j} A$. Since $\Vdash^{S}_{w'} \Box (A \lor \neg A)$ and $w'Rj$ we have $\Vdash^{S}_{j} A \lor \neg A$, thus by Clause 3 either $\Vdash^{S}_{j} A$ or $\Vdash^{S}_{j} \neg A$, hence since $\nVdash^{S}_{j} A$ we conclude $\Vdash^{S}_{j} \neg A$. Since $w'Rj$ and $\Vdash^{S}_{j} \neg A$ we conclude $\Vdash^{S}_{w'} \Diamond \neg A$ by Clause 10, hence by arbitrariness of $w'$ we conclude $\Vdash^{S}_{w} (\Box( A \lor \neg A) \land \neg \Box A) \to (\Diamond \neg A)$ by Clause 4 and by arbitrariness of $w$ and $S$ we conclude $\Vdash^{*} (\Box( A \lor \neg A) \land \neg \Box A) \to (\Diamond \neg A)$ by Clauses 7 and 8.
\end{proof}

\item \textbf{Proposition \ref{prop:IKMKdifferencesecondpart}}

\begin{proof} For the first statement, pick a birelational model $B$ with three worlds $w$, $w'$ and $w''$ such that $w \leq w'$ and $w' R w''$ but such that no other relation holds besides those induced by the properties of $\leq$. It is straightforward to check that this frame satisfies conditions $F_{1}$ and $F_{2}$ (although it does not satisfy $F_{3})$ and that $w \leq w''$ does not hold. We have $w \leq w'$, $w' Rw''$ and $\nVdash^{B}_{w''} \bot$, so we conclude $\nVdash^{B}_{w} \Box \bot$. Likewise, we have $w' \leq w'$, $w' Rw''$ and $\nVdash^{B}_{w''} \bot$, so we conclude $\nVdash^{B}_{w'} \Box \bot$. Now pick any $j \in W$ with $w \leq j$. Then either $j = w''$ or $j = w$, whence $\nVdash^{B}_{j} \Box \bot$, so $\Vdash^{B}_{w} \neg \Box \bot$. Since for no $j \in W$ we have $wRj$ we conclude $\nVdash^{B}_{w} \Diamond \top$, hence $\nVdash^{B}_{w} (\neg \Box \bot) \to (\Diamond \top)$, thus $\nVdash^{B} (\neg \Box \bot) \to (\Diamond \top)$ and $\nVdash (\neg \Box \bot) \to (\Diamond \top)$.  
    
    %For the first statement, pick a birelational model $B$ with three worlds $w$, $w'$ and $w''$ such that $w \leq w'$ and $w' R w''$ but such that no other relation holds besides those induced by the properties of $\leq$. It is straightforward to check that this frame satisfies conditions $F_{1}$ and $F_{2}$ (although it does not satisfy $F_{3})$ and that $w \leq w''$ does not hold. We have $w \leq w'$, $w' Rw''$ and $\nVdash^{B}_{w''} \bot$, so by Clause 9 of Definition \ref{def:validityinbirelational} we have $\nVdash^{B}_{w} \Box \bot$. Likewise, we have $w' \leq w'$, $w' Rw''$ and $\nVdash^{B}_{w''} \bot$, so we have $\nVdash^{B}_{w'} \Box \bot$. Now pick any $j \in W$ with $w \leq j$. Then either $j = w''$ or $j = w$, whence $\nVdash^{B}_{j} \Box \bot$, so $\Vdash^{B}_{w} \neg \Box \bot$ by Clause 4. Since for no $j \in W$ we have $wRj$ we conclude $\nVdash^{B}_{w} \Diamond \top$ by Clause 10, hence $\nVdash^{B}_{w} (\neg \Box \bot) \to (\Diamond \top)$ by Clause 4, thus $\nVdash^{B} (\neg \Box \bot) \to (\Diamond \top)$ by Clause 7 and  $\nVdash (\neg \Box \bot) \to (\Diamond \top)$ by Clause 8.

    For the second statement, let $A$ be an atom $p$. Pick a birelational model $B$ with three worlds $w$, $w'$ and $w''$ and $p \notin v(w'')$ with $w \leq w'$, $w' R w''$ but such that no other relation holds besides those induced by the properties of $\leq$. It is straightforward to check that this frame satisfies conditions $F_{1}$ and $F_{2}$ (although once again it does not satisfy $F_{3})$ and that $w \leq w''$ does not hold. Since $p \notin v(w'')$ we conclude $\nVdash^{B}_{w''} p$. From reflexivity of $\leq$ we have $w'' \leq w''$, and since transitivity does not induce any other relation for $w''$ we conclude $\Vdash^{B}_{w''} \neg p$, so also $\Vdash^{B}_{w''} p \lor \neg p$. Now pick any $w'''$ with $w \leq w'''$.  Then either $w''' = w$ or $w''' = w'$. By construction for no $j \in W$ we have $wRj$ and if $w'Rj$ then $j = w''$, hence since $\Vdash^{B}_{w''} p \lor \neg p$ we conclude $\Vdash^{B}_{w} \Box (p \lor \neg p)$. Since  $w \leq w'$, $w 'Rw''$ and $\nVdash^{B}_{w''} p$ we conclude $\nVdash^{B}_{w} \Box p$, and since $w' \leq w'$, $w 'Rw''$ and $\nVdash^{B}_{w''} p$ we also conclude $\nVdash^{B}_{w'} \Box p$, thence $\Vdash^{B}_{w} \neg \Box p$. Since $\Vdash^{B}_{w} \Box (p \lor \neg p)$ and $\Vdash^{B}_{w} \neg \Box p$ we conclude $\Vdash^{B}_{w} (\Box( p \lor \neg p) \land \neg \Box p)$, but since for no $j \in W$ we have $wRj$ we also conclude $\nVdash^{B}_{w} \Diamond \neg p$, hence $\nVdash^{B}_{w} (\Box( p \lor \neg p) \land \neg \Box p) \to (\Diamond \neg p)$ and thus $\nVdash^{B} (\Box( p \lor \neg p) \land \neg \Box p) \to (\Diamond \neg p)$, whence finally $\nVdash (\Box( p \lor \neg p) \land \neg \Box p) \to (\Diamond \neg p)$.

\end{proof}

\end{itemize}

\subsection{Proofs of Section 4}

\begin{itemize}

\item \textbf{Lemma \ref{lemma:monotonicity}}

\begin{proof}
The proof is just a straightforward adaptation of the proof for intuitionistic logic in the base case, the cases for non-modal connectives and the case with $\Gamma \neq \varnothing$, so we deal only with the modal cases and $\Gamma = \varnothing$.

\begin{enumerate}

        \item ($\Gamma = \varnothing$, $A = \Box B$, partial models): Assume $\vDash^{M, K}_{w} \Box B$. Then, for all $w \leq_{K} w'$, if $K \succ K'$and $w' \in W_{K'}$ then $\vDash^{M, K'}_{w'} B$. Let $w'$ be an arbitrary object with $w \leq_{K} w'$ and $w''$ an arbitrary object with $w' \leq_{K} w''$. Then by transitivity of $\leq_{K}$ we conclude $w \leq_{K} w''$. Now let $K''$ be an arbitrary model with $w'' \in W_{K''}$. Then since $w \leq_{K} w''$, $K \succ K''$ and $w'' \in W_{K''}$ we conclude $\vDash^{M, K''}_{w''} B$. But then for arbitrary $w' \leq_{K} w''$ we have that if $K \succ K''$ and $w'' \in W_{K''}$ then $\vDash^{M, K'}_{w''} B$, so we conclude $\vDash^{M, K}_{w'} \Box B$.

        %Now pick any $w' \leq_{K} w''$, and let $w'''$ be an arbitrary $w'' \leq_{K} w'''$ and $K''$ an arbitrary model with $w''' \in K''$. Then since $w \leq_{K} w''$ and $w''\leq_{K} w'''$ by transitivity of $\leq_{K}$ we conclude $w \leq_{K} w'''$, hence since  $K \succ K''$and $w''' \in K''$ then $\vDash^{K'', M}_{w'''} B$. But then for every $w'' \leq_{K} w'''$ we have that if $K\succ K''$ and $w''' \in K''$ then $\vDash^{K'', M}_{w'''} B$, so we conclude $\vDash^{K,M}_{w'} \Box B$.

    \item  $A = \Box B$, homogeneous models): Assume $\vDash^{H, K}_{w} \Box B$. Then $K \succ K'$ implies $\vDash^{H, K'}_{w} B$. Now pick any $w'$ such that $w \leq w'$ and any $K'$ such that $K \succ K'$ (if any). Since $\vDash^{H, K'}_{w} B$ and $w \leq w'$ the induction hypothesis yields $\vDash^{H, K'}_{w'} B$. But this means that, for arbitrary $K'$, $K \succ K'$ implies $\vDash^{H, K'}_{w'} B$, so we conclude $\vDash^{H, K}_{w'} \Box B$.

    \item ($A = \Diamond B$, partial models): Assume $\vDash^{M, K}_{w} \Diamond B$. Then there is a $K \succ K'$ with $w \in W_{K'}$ and $\vDash^{M, K'}_{w} B$. Now pick any $w'$ such that $w \leq_{K} w'$ and the same $K'$ such that $\vDash^{M, K'}_{w} B$. Since $w \leq_{K} w'$ and $w \in W_{K'}$ by Lemma \ref{lemma:structureofpartialmodels} we have $w \leq_{K'} w'$. Since $\vDash^{M, K'}_{w} B$ and $w \leq_{K'} w'$ the induction hypothesis yields $\vDash^{M, K'}_{w'} B$, and since we already have that $K'$ is a model with $K \succ K'$ and $w' \in W_{K'}$ we conclude $\vDash^{M, K}_{w'} \Diamond B$.

    %From Condition 3 of Then $w \leq_{K'} w'$ must hold in the reference model  be in the reference model of $F_{M}$, so from Clause 2 of Definition \ref{definition:atleastpartiallyhomogeneous} we conclude that $w' \in K'$. Since $\vDash^{K',M}_{w} B$ and $w \leq_{K'} w'$ the induction hypothesis yields $\vDash^{K', M}_{w'} B$, and since we already have that $K'$ is a model with $K \succ K'$ we conclude $\vDash^{K, M}_{w'} \Diamond B$.

    \item  ($A = \Diamond B$, homogeneous models):  Assume $\vDash^{H, K}_{w} \Diamond B$. Then there is a $K \succ K'$ with $\vDash^{H, K'}_{w} B$. Now pick any $w'$ such that $w \leq w'$ and the $K'$ such that $\vDash^{H, K'}_{w} B$. Since $\vDash^{H, K'}_{w} B$ and $w \leq w'$ the induction hypothesis yields $\vDash^{M, K'}_{w'} B$, and since we already have that $K'$ is a model with $K \succ K'$ we conclude $\vDash^{H, M}_{w'} \Diamond B$.

\end{enumerate}

\end{proof}

\end{itemize}

\subsection{Proofs of Section 5}

\begin{itemize}

\item \textbf{Lemma \ref{lemma:existencebirelationamodel}}

\begin{proof}
    We have to show that $M^{*}$ satisfies all properties of birelational models:

    \begin{enumerate}
        \item $\leq^{*}$ is reflexive and transitive;

        \item $\langle w, K \rangle \leq^{*} \langle w', K' \rangle$ implies $v^{*}(\langle w, K \rangle) \subseteq v^{*}(\langle w', K' \rangle)$;

        \item If $\langle w, K \rangle \leq^{*} \langle w', K' \rangle$ and $\langle w, K \rangle R^{*} \langle w'', K'' \rangle$ then there is a unique $\langle w''', K''' \rangle$ with $\langle w'', K'' \rangle \leq^{*} \langle w''', K''' \rangle$ and $\langle w', K' \rangle R^{*} \langle w''', K''' \rangle$

        \item If $\langle w, K \rangle R^{*} \langle w', K' \rangle$ and $\langle w', K' \rangle \leq^{*} \langle w'', K'' \rangle$ then there is a unique $\langle w''', K''' \rangle$ with $\langle w, K \rangle \leq^{*} \langle w''', K''' \rangle$ and $\langle w''', K''' \rangle R^{*} \langle w'', K'' \rangle$

      %  

     %   \item If  $w \leq^{*} w'$ and $w' R^{*} j'$ then there is a unique $j$ such that $w R^{*} j$ and $j \leq^{*} j'$.
    \end{enumerate}

\begin{enumerate}
    \item[(1)] Since for every $w$ and $K$ of $M$ we have $w \leq_{K} w$ clearly $\langle w, K \rangle \leq^{*} \langle w, K \rangle$ for every $\langle w, K \rangle$ by the definition of $\leq^{*}$, so $\leq^{*}$ is reflexive. For transitivity, assume $\langle w, K \rangle \leq^{*} \langle w', K' \rangle$ and $\langle w', K' \rangle \leq^{*} \langle w'', K'' \rangle$. From the definition of $\leq^{*}$ it follows that $K = K' = K''$, $w \leq_{K} w'$ and $w' \leq_{K} w''$. But $\leq_{K}$ is transitive, so we conclude $w \leq_{K} w''$, hence by the definition of $\leq^{*}$ also $\langle w, K \rangle \leq^{*} \langle w'', K'' \rangle$, so $\leq^{*}$ is transitive.

    \item[(2)] Assume $\langle w, K \rangle \leq^{*} \langle w', K' \rangle$. It follows from the definition of $\leq^{*}$ that $K = K'$ and $w \leq_{K} w'$, so from Definition \ref{def:birelational} we get $v_{K}(w) \subseteq v_{K}(w')$. From the definition of $\textit{v}$ we have $v^{*}(\langle w, K \rangle) = v_{K}(w)$ and $v^{*}(\langle w', K' \rangle) = v_{K'}(w')$ and since $K = K'$ also $v^{*}(\langle w', K' \rangle) = v_{K}(w')$, so since $ v_{K}(w) \subseteq  v_{K}(w')$ we have $v^{*}(\langle w, K \rangle) \subseteq v^{*}(\langle w', K' \rangle)$.

    \item[(3)] Assume $\langle w, K \rangle \leq^{*} \langle w', K' \rangle$ and $\langle w, K \rangle R^{*} \langle w'', K'' \rangle$. By definition of $\leq^{*}$ we get $K = K'$ and $w \leq_{K} w'$, and from the definition of $R^{*}$ we get $K \succ K''$ and $w = w''$. Since $K = K'$ and $K \succ K''$ we also get $K' \succ K''$. Since $\langle w, K \rangle R^{*} \langle w'', K'' \rangle$ we have $\langle w'', K'' \rangle \in F^{*}$, so since $w = w''$ we obtain $w \in W_{K''}$ and $K'' \in F$ from the definition of $F^{*}$. Since $w \in W_{K''}$, $w \leq_{K}w'$ and both $K$ and $K''$ are in the same partial model, by Lemma \ref{lemma:structureofpartialmodels} we get $w' \in W_{K''}$ and $w \leq_{K''} w'$. Since $w' \in W_{K''}$ and $K'' \in F$ we conclude $\langle w', K'' \rangle \in F^*$. Now consider the pair $\langle w', K'' \rangle$. Since $w = w''$ and $w \leq_{K''} w'$ we have $w'' \leq_{K''} w'$, so clearly $\langle w'', K'' \rangle \leq^{*} \langle w', K'' \rangle$ by the definition of $\leq^{*}$. Since $K' \succ K''$ we also have $\langle w', K' \rangle R^{*} \langle w', K'' \rangle$ by the definition of $R^{*}$. We conclude that, if $\langle w, K \rangle \leq^{*} \langle w', K' \rangle$ and $\langle w, K \rangle R^{*} \langle w'', K'' \rangle$, then for $\langle w', K'' \rangle$ in particular we have $\langle w'', K'' \rangle \leq^{*} \langle w', K'' \rangle$ and $\langle w', K' \rangle R^{*} \langle w', K'' \rangle$, so the desired witness exists.

    For uniqueness, let $\langle w''', K '''\rangle$ be a pair with  $\langle w'', K'' \rangle \leq^{*} \langle w''', K''' \rangle$ and $\langle w', K' \rangle R^{*} \langle w''', K''' \rangle$. Since  $\langle w'', K'' \rangle \leq^{*} \langle w''', K''' \rangle$ by the definition of $\leq^{*}$ we have $K'' = K'''$, and since $\langle w', K' \rangle R^{*} \langle w''', K''' \rangle$ by the definition of $R^{*}$ we have $w' = w'''$, so $\langle w''', K''' \rangle = \langle w', K'' \rangle$.

    \item[(4)] Assume $\langle w, K \rangle R^{*} \langle w', K' \rangle$ and $\langle w', K' \rangle \leq^{*} \langle w'', K'' \rangle$. Then from the definition of $R^{*}$ we get $K \succ K'$ and $w = w'$, and from the definition of $\leq^{*}$ we get $K' = K''$ and $w' \leq_{K'} w''$. Since $w' \leq_{K'} w''$ and $w = w'$ we get $w \leq_{K'} w''$. Since $K' = K''$ and $K \succ K'$ we also get $K \succ K''$. Since $\langle w, K \rangle R^{*} \langle w', K' \rangle$ then $\langle w, K \rangle \in F^*$, hence by the definition of $F^*$ we have $w \in W_{K}$ and $K \in F$. Since $w \in W_{K}$, $w \leq_{K'} w''$ and both $K$ and $K'$ are in the same partial model, by Lemma \ref{lemma:structureofpartialmodels} we conclude $w'' \in W_{K}$ and $w \leq_{K} w''$. Since $w'' \in W_{K}$ and $K \in F$ we conclude $\langle w'', K \rangle \in F^*$. Now consider the pair $\langle w'', K \rangle$. Since $w \leq_{K} w''$ by the definition of $\leq^{*}$ we have $\langle w, K \rangle \leq^{*} \langle w'', K \rangle$. Since $K \succ K''$ we also have $\langle w'', K \rangle R^{*} \langle w'', K'' \rangle$ by the definition of $R^{*}$.  We conclude that, if  $\langle w, K \rangle R^{*} \langle w', K' \rangle$ and $\langle w', K' \rangle \leq^{*} \langle w'', K'' \rangle$, then for the pair $\langle w'', K \rangle$ in particular we have $\langle w, K \rangle \leq^{*} \langle w'', K \rangle$ and $\langle w'', K \rangle R^{*} \langle w'', K'' \rangle$ so the desired witness exists.

    %Since $\langle w, K \rangle R^{*} \langle w', K' \rangle \in F^*$ and $w = w'$ from the definition of $F^*$ we get $w \in K'$ and $K' \in F$. Since $w \in K'$, $w \leq_{K} w''$ and both $K$ and $K'$ are in the same partial modal model, by Lemma \ref{lemma:structureofpartialmodels} we get $w'' \in K'$ and $w \leq_{K'} w''$

    To show uniqueness, let $\langle w''', K''' \rangle$ be a pair with  $\langle w, K \rangle \leq^{*} \langle w''', K''' \rangle$ and $\langle w''', K''' \rangle R^{*} \langle w'', K'' \rangle$. Since  $\langle w, K \rangle \leq^{*} \langle w''', K''' \rangle$ by the definition of $\leq^{*}$ we have $K = K'''$, and since $\langle w''', K''' \rangle R^{*} \langle w'', K'' \rangle$ by the definition of $R^{*}$ we have $w''' = w''$, so  $\langle w''', K''' \rangle = \langle w'', K \rangle$

%

%\item[(5)] Assume  $\langle w, K \rangle \leq^{*} \langle w', K' \rangle$ and $\langle w', K' \rangle R^{*} \langle w'', K'' \rangle$. Then from the definition of $R^{*}$ we get $K' \succ K''$ and $w' = w''$, and from the definition of $\leq^{*}$ we get $w \leq w'$ and $K = K'$. Since $K = K'$ and $K' \succ K''$ we conclude $K \succ K''$, and since $w\leq w'$ and $w' = w''$ we also conclude $w \leq w''$. Now consider the pair $\langle w, K'' \rangle$, which is guaranteed to exist because $K$ must have the same set of worlds as $K''$ in the modal model $M$. Since $K \succ K''$ we have $\langle w, K \rangle R^{*} \langle w, K'' \rangle$, and since $w \leq w''$ we have $\langle w, K'' \rangle\leq^{*} \langle w'', K''\rangle$, so the desired witness exists.

%

   % To show uniqueness, let $\langle w''', K''' \rangle$ be any pair with  $\langle w, K \rangle \succ ^{*} \langle w''', K''' \rangle$ and $\langle w''', K''' \rangle \leq^{*} \langle w'', K'' \rangle$. Since $\langle w, K \rangle \succ ^{*} \langle w''', K''' \rangle$ by the definition of $R^{*}$ we conclude $w = w'''$, and since $\langle w''', K''' \rangle \leq^{*} \langle w'', K'' \rangle$ by the definition of $\leq^{*}$ we conclude $K'' = K'''$, so we conclude $\langle w''', K''' \rangle = \langle w, K'' \rangle$.

\end{enumerate}

\end{proof}

\item \textbf{Lemma \ref{lemma:correspondenceproofmappingsimilarityintobirelationa}}

   \begin{proof}
        We prove the result by induction on the complexity of $A$, understood as the number of logical symbols occurring on it.

        \begin{enumerate}
            \item $\Gamma = \{ \varnothing\}$:

        \begin{enumerate}
            \item $A = p $ for some atomic $p$. The result follows immediately from the definition of validity for atoms and the fact that $v(\langle w, K \rangle) = v_{K}(w)$ by construction.

            \item $A = B \land C$.  Assume $\vDash^{M,K}_{w} B \land C$. Then $\vDash^{M,K}_{w} B$ and $\vDash^{M,K}_{w} C$. Induction hypothesis: $\Vdash^{M^{*}}_{\langle w, K\rangle} B$ and $\Vdash^{M^{*}}_{\langle w, K\rangle} C$. Then $\Vdash^{M^{*}}_{\langle w, K\rangle} B \land C$. For the converse, assume $\Vdash^{M^{*}}_{\langle w, K\rangle} B \land C$. Then $\Vdash^{M^{*}}_{\langle w, K\rangle} B$ and $\Vdash^{M^{*}}_{\langle w, K\rangle} C$. Induction hypothesis: $\vDash^{M,K}_{w} B$ and $\vDash^{M,K}_{w} C$. Then $\vDash^{M,K}_{w} B \land C$.

            \item $A = B \lor C$.  Assume $\vDash^{M,K}_{w} B \lor C$. Then $\vDash^{M,K}_{w} B$ or $\vDash^{M,K}_{w} C$. Induction hypothesis: $\Vdash^{M^{*}}_{\langle w, K\rangle} B$ or $\Vdash^{M^{*}}_{\langle w, K\rangle} C$. Then $\Vdash^{M^{*}}_{\langle w, K\rangle} B \lor C$. For the converse, assume $\Vdash^{M^{*}}_{\langle w, K\rangle} B \lor C$. Then $\Vdash^{M^{*}}_{\langle w, K\rangle} B$ or $\Vdash^{M^{*}}_{\langle w, K\rangle} C$. Induction hypothesis: $\vDash^{M,K}_{w} B$ or$\vDash^{M,K}_{w} C$. Then $\vDash^{M,K}_{w} B \lor C$.

             \item $A = B \to C$.  Assume $\vDash^{M,K}_{w} B \to C$. Pick any $\langle w', K' \rangle$ with $\langle w, K \rangle \leq^{*} \langle w', K' \rangle$ and $\Vdash^{M^{*}}_{\langle w', K' \rangle} B$. By the definition of $\leq^{*}$ we have $K = K'$ and $w \leq_{K} w'$, so $\Vdash^{M^{*}}_{\langle w', K \rangle} B$.  Induction hypothesis: $\vDash^{M,K}_{w'} B$. Then since $\vDash^{M,K}_{w} B \to C$ and $w \leq_{K} w'$ we conclude $\vDash^{M,K}_{w'} C$, so another use of the induction hypothesis yields $\Vdash^{M^{*}}_{\langle w', K \rangle} C$, and since $K = K'$ also $\Vdash^{M^{*}}_{\langle w', K' \rangle} C$. But then for every $\langle w', K' \rangle$ with  $\langle w, K \rangle \leq^{*} \langle w', K' \rangle$ and $\Vdash^{M^{*}}_{\langle w', K' \rangle} B$ we have $\Vdash^{M^{*}}_{\langle w', K' \rangle} C$, so we conclude $\Vdash^{M^{*}}_{\langle w, K \rangle} B \to C$. For the converse, assume $\Vdash^{M^{*}}_{\langle w, K \rangle} B \to C$. Pick any $w \leq_{K} w'$ with $\vDash^{M,K}_{w'} B$. Induction hypothesis: $\Vdash^{M^{*}}_{\langle w' , K \rangle} B$. Since $w \leq_{K} w'$ by the definition of $\leq^{*}$ we have $\langle w, K \rangle \leq^{*} \langle w', K \rangle$, so since $\Vdash^{M^{*}}_{\langle w, K \rangle} B \to C$ we conclude $\Vdash^{M^{*}}_{\langle w', K \rangle} C$, hence another use of the induction hypothesis yields $\vDash^{M,K}_{w'} C$. But then for every $w \leq_{K} w'$ we have that if $\vDash^{M,K}_{w'} B$ then $\vDash^{M,K}_{w'} C$, hence we conclude $\vDash^{M,K}_{w} B \to C$.

           % \item $A = \Box B$. Assume $\vDash^{M,K}_{w} \Box B$. Then $K \succ K'$ implies  $\vDash^{M,K'}_{w} B$.  Now pick any $\langle w'', K'' \rangle$ with $\langle w, K \rangle \succ^{*} \langle w'', K'' \rangle$. Then from the definition of $\succ^{*}$ we get $K \succ K''$, so we conclude $\vDash^{M,K''}_{w} B$. Induction hypothesis: $\Vdash^{M^{*}}_{\langle w, K''\rangle} B$. But from the definition of $\succ^{*}$ and  $\langle w, K \rangle \succ^{*} \langle w'', K'' \rangle$ we also get $w = w''$, so we conclude $\Vdash^{M^{*}}_{\langle w'', K''\rangle} B$. Then for arbitrary $\langle w'', K'' \rangle$ with $\langle w, K \rangle \succ^{*} \langle w'', K'' \rangle$ we conclude $\Vdash^{M^{*}}_{\langle w'', K''\rangle} B$, so the semantic clause for $\Box$ yields $\Vdash^{M^{*}}_{\langle w, K \rangle} \Box B$. For the converse, assume $\Vdash^{M^{*}}_{\langle w, K \rangle} \Box B$. Then $\langle w, K \rangle \succ^{*} \langle w'', K'' \rangle$ implies $\Vdash^{M^{*}}_{\langle w'', K''\rangle} B$. Now assume $K \succ K'$. Then from the definition of $\succ^{*}$ and the fact that $K$ and $K'$ have the same set of objects we conclude $\langle w, K\rangle \succ^{*}\langle w, K' \rangle$, so we conclude $\Vdash^{M^{*}}_{\langle w, K'\rangle} B$. Induction hypothesis: $\vDash^{M,K'}_{w} B$. But $K'$ is an arbitrary model with $K \succ K'$, so from the semantic clause for $\Box$ we conclude $\vDash^{M,K}_{w} \Box B$.

            \item $A = \Box B$. Assume $\vDash^{M,K}_{w} \Box B$. Then $w \leq_{K} w'$, $K \succ K'$ and $w' \in W_{K'}$ implies  $\vDash^{M,K'}_{w'} B$. Induction hypothesis: $w \leq_{K} w'$, $K \succ K'$ and $w' \in W_{K'}$ implies $\Vdash^{M^{*}}_{\langle w', K'\rangle} B$. Now pick any $\langle w'', K'' \rangle$ with $\langle w, K \rangle \leq^{*} \langle w'', K'' \rangle$ and any $\langle w''', K''' \rangle$ with $\langle w'', K'' \rangle R^{*} \langle w''', K''' \rangle$. From the definition of $R^{*}$ we obtain that $w''' = w''$ and $K'' \succ K'''$, and from the definition of $\leq^{*}$ we obtain that $K = K''$ and $w \leq_{K} w''$. Since $w \leq_{K} w''$ and $w'' = w'''$ we have $w \leq_{K} w'''$. Since $\langle w'', K'' \rangle R^{*} \langle w''', K''' \rangle$ we have $\langle w''', K'''\rangle \in F^*$, so from the definition of $F^*$ we conclude $w''' \in W_{K'''}$. Since $K'' \succ K'''$ and $K = K''$ we conclude $K \succ K'''$. Since $w \leq_{K} w'''$, $K \succ K'''$ and $w''' \in W_{K'''}$ the induction hypothesis yields $\Vdash^{M^{*}}_{\langle w''', K'''\rangle} B$. But then we have that, for arbitrary $\langle w'', K''\rangle$ and $\langle w''', K''' \rangle$, if $\langle w, K \rangle \leq^{*} \langle w'', K'' \rangle$ and $\langle w'', K'' \rangle R^{*} \langle w''', K''' \rangle$ then $\Vdash^{M^{*}}_{\langle w''', K'''\rangle} B$, so we conclude $\Vdash^{M^*}_{\langle w, K \rangle} \Box B$. For the converse, assume $\Vdash^{M^*}_{\langle w, K\rangle} \Box B$. Then, for any $\langle w'', K''\rangle$ and $\langle w''', K''' \rangle$, if $\langle w, K \rangle \leq^{*} \langle w'', K'' \rangle$ and $\langle w'', K'' \rangle R^{*} \langle w''', K''' \rangle$ then $\Vdash^{M^{*}}_{\langle w''', K'''\rangle} B$. Ind. hypothesis: for arbitrary $\langle w'', K''\rangle$ and $\langle w''', K''' \rangle$, if $\langle w, K \rangle \leq^{*} \langle w'', K'' \rangle$ and $\langle w'', K'' \rangle R^{*} \langle w''', K''' \rangle$ then $\Vdash^{M, K'''}_{w'''} B$. Now let $w'$ be an arbitrary object with $w \leq_{K} w'$ and $K'$ an arbitrary propositional model with $K \succ K'$ and $w' \in W_{K'}$. Clearly by the definition of $F^*$ we have $\langle w', K \rangle \in F^*$ and $\langle w', K' \rangle \in F^*$. Since $w \leq_{K} w'$ by the definition of $\leq^{*}$ we have $\langle w, K \rangle \leq^*  \langle w', K \rangle$, and since $K \succ K'$ by the definition of $R^{*}$ we have $\langle w', K \rangle R^{*} \langle w', K' \rangle$. Then since $\langle w, K \rangle \leq^*  \langle w', K \rangle$ and  $\langle w', K \rangle R^{*} \langle w', K' \rangle$ the induction hypothesis yields $\Vdash^{M, K'}_{w'} B$. But then we have that, for arbitrary $w'$ and $K'$, if $w \leq_{K} w'$, $K \succ K'$ and $w' \in W_{K'}$ then $\Vdash^{M, K'}_{w'} B$, so we conclude $\Vdash^{M, K}_{w} \Box B$

            \item $A = \Diamond B$. Assume $\vDash^{M,K}_{w} \Diamond B$. Then there is a $K \succ K'$ such that $w \in W_{K'}$ and $\vDash^{M,K'}_{w} B$. Induction hypothesis: $\Vdash^{M^{*}}_{\langle w, K' \rangle} B$. Since $K \succ K'$ by the definition of $R^{*}$ we get $\langle w, K \rangle R^{*} \langle w, K' \rangle$, so we conclude  $\Vdash^{M^{*}}_{\langle w, K \rangle} \Diamond B$. For the converse, assume $\Vdash^{M^{*}}_{\langle w, K \rangle} \Diamond B$. Then there is some $\langle w', K' \rangle$ such that $\langle w, K \rangle R^{*} \langle w', K' \rangle$ and $\Vdash^{M^{*}}_{\langle w', K' \rangle} B$. From the definition of $R^{*}$ we conclude $w' = w$ and $K \succ K'$, so $\Vdash^{M^{*}}_{\langle w, K' \rangle} B$. Induction hypothesis: $\vDash^{M, K'}_{w} B$. Then since $K \succ K'$ from the semantic clause for $\Diamond$ we get $\vDash^{M, K}_{w} \Diamond B$.

                    \end{enumerate}

      \item $\Gamma \neq \{ \varnothing\}$:

      Assume $\Gamma \vDash^{M, K}_{w} A$. Pick any $\langle w', K' \rangle$ with $\langle w, K \rangle \leq^{*} \langle w', K' \rangle$ and $\Vdash^{M^{*}}_{\langle w', K' \rangle} C$ for all $C \in \Gamma$. From the definition of $\leq^{*}$ we obtain $w \leq_{K} w'$ and $K = K'$, so $\Vdash^{M^{*}}_{\langle w', K \rangle} C$ for all $C \in \Gamma$. From the proof for $\Gamma = \{\varnothing\}$ we get $\vDash^{M, K}_{w'} C$ for all $C \in \Gamma$, hence since $\Gamma \vDash^{M, K}_{w} A$ and $w \leq_{K} w'$ we conclude $\vDash^{M, K}_{w'} A$. Then from the proof for $\Gamma = \{\varnothing\}$ we get $\Vdash^{M^{*}}_{\langle w', K \rangle} A$, and since $K' = K$ also $\vDash^{M^{*}}_{\langle w', K' \rangle} A$. Then for any $\langle w', K' \rangle$ we have that if $\langle w, K \rangle \leq^{*} \langle w', K' \rangle$ and $\vDash^{M^{*}}_{\langle w', K' \rangle} C$ for all $C \in \Gamma$ then $\vDash^{M^{*}}_{\langle w', K' \rangle} A$, so we conclude $\Gamma \Vdash^{M^{*}}_{\langle w, K\rangle} A$. For the converse, assume $\Gamma \Vdash^{M^{*}}_{\langle w, K\rangle} A$, and let $w'$ be any $w \leq_{K} w'$ with  $\vDash^{M, K}_{w'} C$ for all $C \in \Gamma$.  From the proof for $\Gamma = \{\varnothing\}$ we get $\Vdash^{M^{*}}_{\langle w', K \rangle} C$ for all $C \in \Gamma$. Since $w \leq_{K} w'$, from the definition of $\leq^{*}$ we obtain $\langle w, K \rangle \leq^{*} \langle w', K \rangle$, and since $\Gamma \Vdash^{M^{*}}_{\langle w, K\rangle} A$ we conclude $\Vdash^{M^{*}}_{\langle w', K \rangle} A$, hence from the proof for $\Gamma = \{\varnothing \}$ also $\vDash^{M, K}_{w'} A$. Then for any $w \leq w'$ we have that if $\vDash^{M, K}_{w'} C$ for all $C \in \Gamma$ then $\vDash^{M, K}_{w'} A$, so we conclude $\Gamma \vDash^{M, K}_{w} A$.
     
        \end{enumerate}
    \end{proof}

\item \textbf{Lemma \ref{lemma:existencestrongbirelationamodel}}

\begin{proof}
    We have to show that $H^{*}$ satisfies all properties of strong models. Notice that, since $F$ is homogeneous (cf. Definition \ref{def:homogeneouset}), for every $K, K' \in F$ we have $w \in W_{K}$ iff  $w \in W_{K'}$ and $w \leq_{K} w'$ iff $w \leq_{K'} w'$ for all $w$ and $w'$, so relational subscripts can be dropped altogether.  Notice also that, since all homogeneous models are also partial models (cf. Definitions \ref{def:partialmodel} and \ref{def:homogeneousmodel}), satisfaction of all properties of birelational models already follows Lemma \ref{lemma:existencebirelationamodel}, which means that we only have to show the following:

    \begin{enumerate}
        \item If $\langle w, K \rangle \leq^{*} \langle w', K' \rangle$ and $\langle w', K' \rangle R^{*} \langle w'', K'' \rangle$ then there is a unique $\langle w''', K''' \rangle$ such that $\langle w, K \rangle R^{*} \langle w''', K''' \rangle$ and $\langle w''', K''' \rangle \leq^{*} \langle w'', K'' \rangle$.

\begin{enumerate}
    \item  Assume $\langle w, K \rangle \leq^{*} \langle w', K' \rangle$ and $\langle w', K' \rangle R^{*} \langle w'', K'' \rangle$.  Then from the definition of $\leq^{*}$ we get $K = K'$ and $w \leq w'$, and from the definition of $R^{*}$ we get $K' \succ K''$ and $w' = w''$. Since we have $\langle w, K \rangle \leq^{*} \langle w', K' \rangle$ and $\langle w', K' \rangle R^{*} \langle w'', K'' \rangle$ we have both $\langle w, K \rangle \in F^{*}$ and $\langle w'', K'' \rangle \in F^{*}$, so also $w \in W_{K}$, $w'' \in W_{K''}$, $K \in F$ and $K'' \in F$. Since both $K$ and $K''$ are elements of $F$, $w \in W_{K}$ and $F$ is homogeneous we conclude $w \in W_{K''}$, so $\langle w, K'' \rangle \in F^{*}$. Now consider the pair $\langle w, K'' \rangle$. Since $K = K'$ and $K' \succ K''$ we conclude $K \succ K''$, so  $\langle w, K \rangle R^{*} \langle w, K'' \rangle$ by the definition of $R^{*}$. Since $w \leq w'$ and $w' = w''$ we have $w \leq w''$, so we conclude $\langle w, K'' \rangle \leq^{*} \langle w'', K'' \rangle$ by the definition of $\leq^{*}$, hence the desired witness exists.

  For uniqueness, pick a pair $\langle w''', K '''\rangle$ such that  $\langle w, K \rangle R^{*} \langle w''', K''' \rangle$ and $\langle w''', K''' \rangle \leq^{*} \langle w'', K'' \rangle$. Since  $\langle w, K \rangle R^{*} \langle w''', K''' \rangle$ by the definition of $R^{*}$ we have $w = w'''$, and since $\langle w''', K''' \rangle \leq^{*} \langle w'', K'' \rangle$ by the definition of $\leq^{*}$ we have $K''' = K''$, so $\langle w''', K''' \rangle = \langle w, K'' \rangle$.
\end{enumerate}
    
\end{enumerate}

\end{proof}

\item \textbf{Proposition \ref{prop:excessivemodels}}

\begin{proof}
Assume $\langle w, K \rangle \leq^{*} \langle w', K' \rangle$ and $\langle w'', K'' \rangle R^{*} \langle w', K' \rangle$.  Then from the definition of $\leq^{*}$ we get $K = K'$ and $w \leq w'$, and from the definition of $R^{*}$ we get $K'' \succ K'$ and $w' = w''$. Since we have $\langle w, K \rangle \leq^{*} \langle w', K' \rangle$ and $\langle w'', K'' \rangle R^{*} \langle w', K' \rangle$   $\langle w', K'' \rangle$ we also have both $\langle w, K \rangle \in F^{*}$ and $\langle w'', K'' \rangle \in F^{*}$, so also $w \in W_{K}$, $K \in F$ and $K'' \in F$. Since both $K$ and $K''$ are elements of $F$, $w \in W_{K}$ and $F$ is homogeneous we conclude $w \in W_{K''}$, so also $\langle w, K'' \rangle \in F^{*}$. Now consider the pair $\langle w, K'' \rangle$. Since $K = K'$ and $K'' \succ K'$ we have $K'' \succ K$, so we conclude $\langle w, K'' \rangle R^{*} \langle w, K \rangle$ by the definition of $R^{*}$. Since $w \leq w'$ and $w' = w''$ we conclude $w \leq w''$, hence $\langle w', K'' \rangle \leq^{*} \langle w'', K'' \rangle$ by the definition of $\leq^{*}$, so the desired witness exists.

    To show uniqueness, pick a $\langle w''', K '''\rangle$ for which it holds that $\langle w''', K''' \rangle R^{*} \langle w, K \rangle$ and $\langle w''', K''' \rangle \leq^{*} \langle w'', K'' \rangle$. Since  $\langle w''', K''' \rangle R^{*} \langle w, K \rangle$ by the definition of $R^{*}$ we have $w = w'''$, and since $\langle w''', K''' \rangle \leq^{*} \langle w'', K'' \rangle$ by the definition of $\leq^{*}$ we have $K''' = K''$, so we conclude $\langle w''', K''' \rangle = \langle w, K'' \rangle$.

\end{proof}

\item \textbf{Lemma \ref{lemma:correspondenceproofmappingshomogeneous}}

\begin{proof}
    The proof is identical to that of Lemma \ref{lemma:correspondenceproofmappingsimilarityintobirelationa} in all cases except when $\Gamma = \varnothing$ and $A$ is either $\Box B$ or $\Diamond B$, so we deal only with those.

    \begin{enumerate}
        \item $A = \Box B$. Assume $\vDash^{H,K}_{w} \Box B$. Then $K \succ K'$ implies $\vDash^{H,K'}_{w} B$. Induction hypothesis: $K \succ K'$ implies $\Vdash^{H^{*}}_{\langle w , K' \rangle} B$. Now pick any $\langle w', K' \rangle$ with $\langle w, K \rangle R^{*} \langle w', K' \rangle$. By the definition of $R^{*}$ we have $w = w'$ and $K \succ K'$, so the induction hypothesis yields $\Vdash^{H^{*}}_{\langle w', K' \rangle} B$, and since we picked an arbitrary $\langle w', K' \rangle$ with $\langle w, K \rangle R^{*} \langle w', K' \rangle$ we conclude  $\Vdash^{H^{*}}_{\langle w , K \rangle} \Box B$. For the converse, assume $\Vdash^{H^{*}}_{\langle w , K \rangle} \Box B$. Then  $\langle w, K \rangle R^{*} \langle w', K' \rangle$ implies $\Vdash^{H^{*}}_{\langle w', K' \rangle} B$. Since $\langle w, K \rangle R^{*} \langle w', K' \rangle$ implies $w = w'$ due to how $R^{*}$ is defined, we also conclude that $\langle w, K \rangle R^{*} \langle w', K' \rangle$ implies $\Vdash^{H^{*}}_{\langle w, K' \rangle} B$. Induction hypothesis: $\langle w, K \rangle R^{*} \langle w', K' \rangle$ implies $\vDash^{H, K'}_{w} B$. Now pick any $K''$ such that $K \succ K''$. Since $w \in W_{K}$ and $H$ is a homogeneous model we conclude $w \in W_{K''}$, so also $\langle w, K'' \rangle \in F^*$. Since $K \succ K''$ we have $\langle w, K \rangle R^{*} \langle w, K'' \rangle$ by the definition of $R^{*}$, so the induction hypothesis yields $\vDash^{H, K''}_{w} B$. Since $K''$ was an arbitrary model with $K \succ K''$ we conclude $\Vdash^{H, K}_{w} \Box B$.

        \item $A = \Diamond B$. Assume $\vDash^{H,K}_{w} \Diamond B$. Then there is a $K \succ K'$ with $\vDash^{H, K'}_{w} B$. Induction hypothesis: there is a $K \succ K'$ with $\Vdash^{H^{*}}_{\langle w, K' \rangle} B$. Since $K \succ K'$ we have $\langle w, K \rangle R^{*} \langle w, K' \rangle$, so we already conclude $\Vdash^{H^{*}}_{\langle w, K \rangle} \Diamond B$. For the converse, assume $\Vdash^{H^{*}}_{\langle w, K \rangle} \Diamond B$. Then there is $\langle w, K \rangle R^{*} \langle w', K' \rangle$ such that $\Vdash^{H^{*}}_{\langle w', K' \rangle} B$. Since $\langle w, K \rangle R^{*} \langle w', K' \rangle$ implies $w = w'$ we can conclude $\Vdash^{H^{*}}_{\langle w, K' \rangle} B$. Induction hypothesis: $\vDash^{H, K'}_{w} B$. Since $\langle w, K \rangle R^{*} \langle w', K' \rangle$ also implies $K \succ K'$ we conclude $\vDash^{H, K}_{w} \Diamond B$.
         
    \end{enumerate}

\end{proof}

\end{itemize}

\subsection{Proofs of Section 6}

\begin{itemize}

\item \textbf{Lemma \ref{lemma:atleastpartialcopyFC}}

\begin{proof}
   We need to prove that $PC^j$ satisfies all three conditions of Definition \ref{def:partialcopy} with respect to $FC$.

   \begin{enumerate}
       \item  Since $Q_{PC^j} = \{ j' | j \leq_{PC} j'\}$ and $\leq_{PC}$ is a relation on the objects of $Q_{PC}$ clearly $Q_{PC^j} \subseteq Q_{PC}$, so since $Q_{PC} \subseteq Q_{FC}$ by transitivity of inclusion we have $Q_{PC^j} \subseteq Q_{FC}$;

       \item Assume $j' \leq_{FC} j''$ and $j' \in Q_{PC^j}$. From the definition of $Q_{PC^j}$ we get $j \leq_{PC} j'$. From Clause 3 of Definition \ref{def:partialcopy} we get $j \leq_{FC} j'$, so since $j' \leq_{FC} j''$ by transitivity of $\leq_{FC}$ we get $j \leq_{FC} j''$. Then since $j \in Q_{PC}$ and $j \leq_{FC} j''$ from Clause 2 of Definition \ref{def:partialcopy} we get $j'' \in Q_{PC}$, and since $j \leq_{FC} j''$ from Clause 3 of Definition \ref{def:partialcopy} we get $j \leq_{PC} j''$, so from the definition of $Q_{PC^j}$ we conclude $j'' \in Q_{PC_{j}}$.

       \item From the definition of $\leq_{PC^J}$ we have that, for all  $j', j '' \in Q_{PC^j}$, $j' \leq_{PC^j} j''$ iff $j' \leq_{PC} j''$. From Clause 3 of Definition \ref{def:partialcopy} and the fact that $Q_{PC^J} \subseteq Q_{PC}$ we also have $j' \leq_{PC} j''$ iff $j' \leq_{FC} j''$ for all $j', j'' \in Q_{PC^j}$, hence $j', j''' \in Q_{PC^j}$ implies $j' \leq_{PC^j} j''$ iff $j' \leq_{PC} j''$ iff $j' \leq_{FC} j''$. 
   \end{enumerate}

\end{proof}

\item \textbf{Lemma \ref{lemma:firstfunctionexistence}}

\begin{proof}
  Let $\langle Q, R \rangle$ be the canonical semi-structure from which $FC$ is obtained. Let $n$ be the smallest number such that there is a $j$ with $j \in Q_{PC}$ and $j \in Q_n$ for the $Q_{n}$ that make up $Q$. Since from the definition of $R$ we have that $jRj'$ and $j \in Q_{m }$ implies $j' \in Q_{m +1}$ clearly for no $j' \in Q_{i}$ ($i \leq n$) we have $j'Rj$, and since $\leq_{PC}$ is the transitive and reflexive closure of $R$ we have $j' \leq_{PC} j$ implies $j' = j$.

  For every $w' \in W$, let $e_{w'}$ be a surjective function from the natural numbers to the set $w'_{\leq} = \{ w' | w \leq w'\}$, the existence of which can be guaranteed because $W$ is countable. For every $j \in Q_{PC}$ with $j \in Q_{m} (m > 0)$, let $e_{j}$ be a bijection from the set $j_{R} = \{ j' |j R j' \}$ to the natural numbers, the existence of which can be guaranteed because by definition there are (countably) infinitely many such $j'$ for each $j$. Then our function $f_{B,w}$ is defined as follows, in which $n$ is still the smallest number such that there is a $j$ with $j \in Q_{PC}$ and $j \in Q_n$ for the $Q_{n}$ that make up $Q$:

  \begin{enumerate}

      \item For all $j \in Q_{PC}$ with $j \in Q_{n}$, $f_{B,w}(j) = v(w)$;

     \item  For all $j' \in Q_{PC}$ with $j \in Q_{m} (m > n)$, let $j$ be the singular object with $j Rj'$. Then if $f_{B, w}(j) = v(w')$, $e_{w'}(i) = w''$ and $e_{j}(j') = i$ then $f_{B, w}(j') = v(w'')$.

  \end{enumerate}

So the function $e_{w}$ assigns to all numbers some extension of $w$, the function $e_{j}$ assigns to all dedicated extensions of $j$ some unique number and the function $f_{B,w}$ simply assigns to each extension $j'$ of $j$ the extension $w'$ of $w$ with the same number, thus guaranteeing that the atomic assignments of every extension $w'$ will be assigned to some $j'$.  Now we prove that $f_{B,w}$ indeed satisfies all requirements of a interpretation function (see Definition \ref{def:interpretationfunction}). The proof proceeds via induction on the numbering $m$ of sets $Q_m$ such that there is a $j$ with $j \in Q_{PC}$ and $j \in Q_{n}$.

\begin{enumerate}

\item There is some $j \in Q_{PC}$ such that $f_{B,w}(j) = v(w)$;

 Since $w$ is assigned to each object of the set $Q_{n}$ with the smallest $n$ this follows immediately from the definition.

    \item For all $j \in Q_{PC}$, $f_{B, w}(j) = v(w')$ for some $w \leq w'$.

    (Base case) Since for all $j \in Q_{PC}$ with $j \in Q_{n}$ we have $f_{B,w}(j) = v(w)$ and $w \leq w$ we conclude that for all such $j$ we have $v(j) = v(w')$ for some $w \leq w'$.

    (Inductive step) Pick any $j' \in Q_{PC}$ with $j' \in Q_{m} (m > n)$. Let $j$ be the singular $j$ with $jRj'$. Since $j \in Q_{m-1}$ the induction hypothesis yields $f_{B, w}(j) = v(w')$ for some $w \leq w'$. Then by the definition of $e_{w'}$ and $f_{B,w}$ we have $f_{B, w}(j') = v(w'')$ for some $w' \leq w''$. Since $w \leq w'$ and $w' \leq w''$ by transitivity of $\leq$ we conclude $w \leq w''$, so for this arbitrary $j'$ we have $f_{B,w}(j') = v(w'')$ for some $w \leq w''$.

    \item For all $j \in Q_{PC}$, if $f_{B, w}(j) = v(w')$ and $w' \leq w''$, then there is a $j' \in Q_{PC}$ such that $j \leq_{PC} j'$ and $f_{B, w}(j') = v(w'')$.

    (Inductive step and base case) Pick any $j$ with $j \in Q_{PC}$ and $j \in Q_{m} ( n \leq m)$. Without loss of generality, let $f_{B, w}(j) = v(w')$. Pick any $w'' \in W$ with $w' \leq w''$. Since $e_{w'}$ is a surjection from the natural numbers to $w'_{\leq}$ there is some number $i$ with $f_{B, w}(i) = w''$. Since $e_{j}$ is a bijection from the natural numbers to $j_{R}$ there is some $j'$ with $e_{j}(j') = i$ and $jRj'$, and from the definition of $\leq_{PC}$ we also get $j \leq_{PC} j'$. Then from the definition of $f_{B, w}$ we conclude $f_{B, w}(j') = v(w'')$, so $j'$ witnesses the desired property for the arbitrary $w''$ with $w' \leq w''$.

    \item For all $j,j' \in Q_{PC}$, if $f_{B, w}(j) = v(w')$, $f_{B, w}(j') = v(w'')$ and $j \leq_{PC} j'$ then $w' \leq w''$.

    (Base case) As shown at the beginning of this proof, for every $j' \in Q_{PC}$ with $j' \in Q_{n}$ we have that $j \leq j'$ implies $j = j'$, so clearly for all $j \leq j'$ we have $f_{B,w} (j) = v(w)$, and since $w \leq w$ the desired property follows immediately.

    (Inductive step) Pick any $j'$ with $j' \in Q_{PC}$ and $j' \in Q_{m} (m > n)$. Pick any $j$ with $j \leq_{PC} j'$. Since $\leq_{PC}$ is the reflexive and transitive closure of $R$ we have that either $j = j'$, $jRj'$ or there is a chain $jR j^{0} R \ldots R j^{n} Rj' $ for some $n \geq 0$. If $j = j'$ then $f_{B,w}(j) = f_{B,w}(j') = v(w')$ for some $w' \in W$, and the result follows from the fact that $w' \leq w'$. If $jRj'$ and $f_{B,w}(j) = v(w')$ then by the very definition of $e_{w'}$ and $f_{B, w}$ we have $f_{B, w}(j') = v(w'')$ for some $w' \leq w''$, as desired. If there is chain $jR j^{0} R \ldots R j^{n} Rj' $, let $f_{B,w}(j) = v(w')$. Then from the definition of $e_{w'}$ and $f_{B,w}$ we get that $f_{B,w}(j^{0}) = v(w^{0})$ for some $w' \leq w^{0}$, for every $0 < m \leq n$ we get that $f_{B,w}(j^{m}) = v(w^{m})$ for some $w^{m-1} \leq w^{m}$ and also $f_{B,w}(j') = v(w'')$ for some $w^{n} \leq w^{''}$, so from the transitivity of $\leq$ we conclude $w' \leq w''$.

\end{enumerate}
      
\end{proof}

\item \textbf{Lemma \ref{lemma:proofofequalityforworlds}}

    \begin{proof}
        We prove the result by induction on the complexity of $A$.

        \begin{enumerate}
            \item $\Gamma = \{ \varnothing\}$:

        \begin{enumerate}
            \item $A = p $ for atomic $p$. The result follows immediately from the clauses for atomic validity in Definitions \ref{def:validityinpartialsimilaritymodel} and \ref{def:validityinbirelational} and the fact that $f_{B,w}(j) = v(w')$.

            \item $A = B \land C$.  Assume $\Gamma \vDash_{j}^{M^B, K} B \land C$. Then $ \vDash_{j}^{M^B, K} B$ and $ \vDash_{j}^{M^B, K} C$. Induction hypothesis: $\Vdash^{B}_{w'} B$ and $\Vdash^{B}_{w'} C$. Then $\Vdash^{B}_{w'} B \land C$. For the converse, assume $\Vdash^{B}_{w'} B \land C$. Then $\Vdash^{B}_{w'} B$ and $\Vdash^{B}_{w'} C$. Induction hypothesis: $ \vDash_{j}^{M^B, K} B$ and $ \vDash_{j}^{M^B, K} C$. Then $ \vDash_{j}^{M^B, K} B \land C$.

            \item $A = B \lor C$.  Assume $ \vDash_{j}^{M^B, K} B \lor C$. Then $ \vDash_{j}^{M^B, K} B$ or $ \vDash_{j}^{M^B, K} C$. Induction hypothesis: $\Vdash^{B}_{w'} B$ or $\Vdash^{B}_{w'} C$. Then $\Vdash^{B}_{w'} B \lor C$. For the converse, assume $\Vdash^{B}_{w'} B \lor C$. Then $\Vdash^{B}_{w'} B$ or $\Vdash^{B}_{w'} C$. Induction hypothesis: $ \vDash_{j}^{M^B, K} B$ or $ \vDash_{j}^{M^B, K} C$. Then $ \vDash_{j}^{M^B, K} B \lor C$.

             \item $A = B \to C$.  Assume $ \vDash_{j}^{M^B, K} B \to C$. Pick any $w' \leq w''$ with $\Vdash^{B}_{w''} B$. From Clause 3 of Definition \ref{def:interpretationfunction} we have that, since $f_{B,w}(j) = v(w')$ and $w' \leq w''$, there must be some $j \leq_{PC} j'$ such that $f_{B,w}(j') = v(w'')$. Since $f_{B,w}(j') = v(w'')$ and $\Vdash^{B}_{w''} B$ the induction hypothesis yields  $ \vDash_{j'}^{M^B, K} B$, and since  $ \vDash_{j}^{M^B, K} B \to C$ and $j \leq_{PC} j'$ by the semantic clause for implication we conclude $ \vDash_{j'}^{M^B, K} C$. But then since $f_{B,w}(j') = v(w'')$ we can apply the induction hypothesis again to conclude $\Vdash^{B}_{w''} C$. Hence for all $w''$ with $w' \leq w''$ we have that if $\Vdash^{B}_{w''} B$ then $\Vdash^{B}_{w''} C$, whence we conclude $\Vdash^{B}_{w'} B \to C$. For the converse, assume $\Vdash^{B}_{w'} B \to C$. Pick any $j \leq_{PC} j'$ with $ \vDash_{j'}^{M^B, K} B$. Without loss of generality, let $f_{B,w}(j') = v(w'')$. The induction hypothesis yields $\Vdash^{B}_{w''} B$. Since $f_{B,w}(j') = v(w'')$ and $j \leq_{PC} j'$ by Clause 4 of Definition \ref{def:interpretationfunction} we have $w' \leq w''$. Since $\Vdash^{B}_{w'} B \to C$, $w' \leq w''$ and $\Vdash^{B}_{w''} B$ the semantic clause for implication yields $\Vdash^{B}_{w''} C$, and the induction hypothesis yields $ \vDash_{j'}^{M^B, K} C$. Then since $j'$ was arbitrary we have that, for every $j \leq_{PC} j'$, if $ \vDash_{j'}^{M^B, K} B$ then $ \vDash_{j'}^{M^B, K} C$, so the semantic clause for implication yields $ \vDash_{j}^{M^B, K} B \to C$.

            \item $A = \Box B$. Assume $ \vDash_{j}^{M^B, K} \Box B$. Then $j \leq_{PC} j'$, $K \succ K'$ and $j' \in Q_{K'}$ implies  $\vDash^{M^B,K'}_{j'} B$. Pick any $w''$ and $w'''$ such that $w' \leq w''$ and $w'' R w'''$. Since $w' \leq w''$, by Clause 3 of Definition \ref{def:interpretationfunction} we conclude that there is some $j \leq_{PC} j'$ such that $f_{B,w}(j') = v(w'')$. Since $f_{B,w}(j') = v(w'')$, $w' \leq w''$ and $w'' R w'''$ by Lemma \ref{lemma:modalexistence} we conclude that there is a frame $K' = \langle Q_{PC^{j'}}, \leq_{PC^{j'}}, f_{B,w'''} \rangle$ such that if $j'' \in Q_{PC}$, $j'' \in Q_{PC^{j'}}$, $f_{B,w}(j'') = v(s)$ and $f_{B,w'''}(j'') = v(s')$ then $s R s'$. By Definition \ref{def:equivalentbirelationalmodel} this means that $K \succ K'$, and by Definition \ref{def:PCj} we have $j' \in Q_{PC^{j'}}$. But then since $j \leq_{PC} j'$, $j' \in Q_{PC_{j'}}$ and $K \succ K'$ we conclude $\vDash^{M^B,K'}_{j'} B$. Since by the definition of $f_{B,w'''}$ in Lemma \ref{lemma:modalexistence} we have $f_{B,w'''}(j') = v(w''')$ the induction hypothesis yields $\Vdash^{B}_{w'''} B$. But $w'''$ is an arbitrary world with $w'' R w'''$ for arbitrary $w''$ with $w' \leq w''$, so by the clause for necessity we conclude $\Vdash^{B}_{w'} \Box B$. For the converse, assume $\Vdash^{B}_{w'} \Box B$. Then, for every $w''$ and $w'''$, if $w' \leq w''$ and $w'' R w'''$ then $\Vdash^{B}_{w'''} B$. Now pick any $j \leq_{PC} j'$ and any $K' = \langle Q'_{PC'}, \leq'_{PC'} f_{B, s} \rangle$ with $K \succ K'$ and $j' \in Q^{'}_{PC'}$. Then since $K \succ K'$ from Definition \ref{def:equivalentbirelationalmodel} we conclude that if $j'' \in Q_{PC}$, $j'' \in Q^{'}_{PC'}$, $f_{B, w'}(j'') = v(s')$ and $f_{B, s}(j'') = v(s'')$ then $s' R s''$. Let $f_{B,w} (j') = v(w'')$ and $f_{B,s} (j') = v(w''')$ for arbitrary $w''$ and $w'''$. Then since $j' \in Q_{PC}$ and $j' \in Q'_{PC'}$ we have $w'' R w'''$. Since $f_{B,w}(j') = v(w'')$, $f_{B,w}(j) = v(w')$ and $j \leq_{PC} j'$ by  by Clause 4 of Definition \ref{def:interpretationfunction} we have $w' \leq w''$. Since $w' \leq w''$ and $w'' R w'''$ we conclude $\Vdash^{B}_{w'''} B$, so since $f_{B,s}(j') = v(w''')$ the induction hypothesis yields $ \vDash_{j'}^{M^B, K'} B$. But $j'$ is an arbitrary object and $K'$ an arbitrary frame with $j \leq_{PC} j'$, $j' \in K'$ and $K \succ K'$, so  we conclude $ \vDash_{j}^{M^B, K} \Box B$.

            \item $A = \Diamond B$. Assume $\vDash_{j}^{M^B, K} \Diamond B$. Then there is a $K \succ K'$ with $j \in K'$ and $ \vDash_{j}^{M^B, K'} B$. without loss of generality, let $K' = \{Q'_{P C'}, \leq_{PC'}, f_{B,s}\}$ and $f_{B,s}(j) = v(w'')$. Then since $f_{B,w}(j) = v(w')$ from the definition of $\succ$ we have $w' R w''$. Since $ \vDash_{j}^{M^B, K'} B$ and $f_{B,s}(j) = v(w'')$ the induction hypothesis yields $\Vdash^{B}_{w''} B$. Since $w' R w''$ and $\Vdash^{B}_{w''} B$ we can use the semantic clause for possibility to conclude $\Vdash^{B}_{w'} \Diamond B$. For the converse, assume $\Vdash^{B}_{w'} \Diamond B$. Then there is a $w''$ with $w' R w''$ and $\Vdash^{B}_{w''} B$. Since $f_{B,w}(j) = v(w')$, by Clause 2 of Definition \ref{def:interpretationfunction} we have $w \leq w'$. Since $w \leq w'$ and $w' R w''$ by Lemma \ref{lemma:modalexistence} we conclude that there is a $K' = \langle Q_{PC^j}, \leq_{PC^j}, f_{B,w''} \rangle$ such that if $j' \in Q_{PC}$, $j' \in Q_{PC^{j}}$, $f_{B,w}(j') = v(s)$ and $f_{B,w''}(j') = v(s')$ then $s R s'$. By the definition of $\succ$ we get $K \succ K'$, and by Definition \ref{def:PCj} we also have $j \in Q_{PC_{j}}$. By the definition of $f_{B,w''}$ in Lemma \ref{lemma:modalexistence} we also have $f_{B,w''} (j) = v(w'')$, so the induction hypothesis yields $ \vDash_{j}^{M^B, K'}  B$. Hence we have $K \succ K'$, $j \in Q_{PC^j}$ and $ \vDash_{j}^{M^B, K'}  B$, so we can use the semantic clause for possibility to conclude $ \vDash_{j}^{M^B, K} \Diamond B$.

                    \end{enumerate}

      \item $\Gamma \neq \{ \varnothing\}$:

      Assume $ \Gamma \vDash_{j}^{M^B, K} A$. Pick any $w''$ with $w' \leq w''$ and $\Vdash^{B}_{w''} B$ for every $B \in \Gamma$. Since $w' \leq w''$ by Clause 3 of Definition \ref{def:interpretationfunction} there must be some $j \leq_{PC} j'$ with $f_{B,w}(j') = v(w'')$. Fix any such $j'$. Since $f_{B,w}(j') = v(w'')$ and $\Vdash^{B}_{w''} B$ for every $B \in \Gamma$ our proof for $\Gamma = \{\varnothing\}$ shows that $\vDash^{M^B, K}_{j'} B$ holds for every $B \in \Gamma$, and since $ \Gamma \vDash_{j}^{M^B, K} A$ and $j \leq_{PC} j'$ we conclude $ \vDash_{j'}^{M^B, K} A$. Since $ \vDash_{j'}^{M^B, K} A$ and $f_{B,w}(j') = v(w'')$ from the proof for $\Gamma = \{\varnothing\}$ we conclude $\Vdash^{B}_{w''} A$. But then we have that, for every $w''$ with $w' \leq w''$, if $\Vdash^{B}_{w''} B$ for every $B \in \Gamma$ then $\Vdash^{B}_{w''} A$, so we conclude $\Gamma \Vdash^{B}_{w'} A$. For the converse, assume $\Gamma \Vdash^{B}_{w'} A$. Pick any $j'$ with $j \leq_{PC} j'$ and $\vDash^{M^B, K}_{j'} B$ for every $B \in \Gamma$. Without loss of generality, let $f_{B,w}(j') = v(w'')$. Since $f_{B,w}(j') = v(w'')$, $f_{B,w}(j) = v(w')$ and $j \leq_{PC} j'$ by Clause 3 of Definition \ref{minimalpropositionalmodels} we conclude $w' \leq w''$. Since $f_{B,w}(j') = v(w'')$ and $\vDash^{M^B, K}_{j'} B$ for every $B \in \Gamma$ our proof for $\Gamma = \{ \varnothing\}$ yields $\Vdash^{B}_{w''} B$ for every $B \in \Gamma$, and since $\Gamma \Vdash^{B}_{w'} A$ and $w' \leq w''$ we conclude $\Vdash^{B}_{w''} A$. Since $\Vdash^{B}_{w''} A$ and $f_{B,w}(j') = v(w'')$  our proof for $\Gamma = \{ \varnothing\}$ yields $\vDash^{M^B, K}_{j'} A$. But then for every $j'$ such that $j \leq_{PC} j'$ we have that if $\vDash^{M^B, K}_{j'} B$ for every $B \in \Gamma$ then $\vDash^{M^B, K}_{j'} A$, so we conclude $ \Gamma \vDash_{j}^{M^B, K} A$.

        \end{enumerate}
    \end{proof}

\item \textbf{Lemma \ref{lemma:modalexistencehomogeneous}}

\begin{proof}

%Notice that, by Clause 4 of Definition \ref{def:interpretationfunction}, if $f_{B,w}(j') = v(s)$ for some $j' \in PC^{j}$ then since $j \leq_{PC} j'$ we also have $w' \leq s$. With this in mind, 

Let us remember that, in virtue of Clause 2 of Definition \ref{def:interpretationfunction}, $f_{S,w} (j') = v(s)$ implies $w \leq s$. Then we construct $f_{S,w'''}$ as follows:

%$j \leq_{PC}^* j'$ for any We define the function $f_{w'', B}$ as follows:

%First notice that if $f_{B,w} (j) = v(w''')$ then $w \leq w'''$ by Clause $1$ of the definition of $f_{B, w}$. With this in mind, we define the function $f_{w'', B}$ as follows:

\begin{enumerate}
    \item Since $f_{S,w}(j) = v(w')$ we conclude $w \leq w'$. Since $w' Rw''$, by $F_{3}$ and witness uniqueness there must be a unique $w''''$ such that $wRw''''$ and $w'''' \leq w''$. We put $w''' = w''''$ and $f_{S, w'''}(j) = v(w''')$ for all $j \in Q_{0}$.

\item For all $j' \in Q_{n}$ for $n > 0$, let $f_{S, w}(j') = v (s)$. Then $w \leq s$ and, since $w' R w'''$, from condition $F_{1}$ and witness uniqueness we conclude that there must be a unique $s'$ such that $s R s'$ and $w''' \leq s'$. Then we define $f_{S,w'''}(j) = v(s')$.

%From Clause 1 Definition \ref{def:interpretationfunction} we have $w \leq s$. Since $w \leq s$ and $w R$ Since $w \leq $ Then we define $f_{B, w''} (j') = v(s')$,  where $s'$ is the unique object such that $s R s'$ and $w'' \leq s'$;
\end{enumerate}

Since by definition $f_{S, w'''}(j) = v(w''')$ the function clearly satisfies Clause 1 of Definition \ref{def:interpretationfunction}, and since by definition for all $j' \in Q_{FC}$ we have $f_{S, w'''}(j') = v(s')$ for some $w''' \leq s'$ it clearly also satisfies Clause 2. 

We have to show that $f_{S, w'''}(j) = v(w'')$. For that, let $f_{S, w'''}(j) = v(w'''')$. Since $f_{S, w'''}$ satisfies Clause 2 of Definition \ref{def:interpretationfunction}, $w''' \leq w''''$. By the definition of $f_{S, w'''}$ and the fact that $f_{S,w}(j) = v(w')$ we also have $w'Rw''''$ and $w \leq w'$. In the construction of $f_{S,w'''}$ above we have shown that $w''' \leq w''$, and it is assumed in the statement of the lemma that $w'Rw''$. Then $w''' \leq w''$, $w'Rw''$, $w''' \leq w''''$ and $w'Rw''''$, so by uniqueness of $w''''$ we have $w'''' = w''$, thus we conclude $f_{S, w'''}(j) = v(w'')$.

 In order to finish the proof and show that $f_{S,w'''}$ is indeed an interpretation of $w''' \in W$ in $FC$ it suffices to show that it satisfies the two remaining clauses of Definition \ref{def:interpretationfunction}. The proof is essentially the same of Lemma \ref{lemma:modalexistence}, the sole difference being that the homogenous character of $FC$ allows us to drop all references to Lemma \ref{def:PCj} and to structural features of $PC^j$ (and thus simplifies the proof).

\begin{enumerate}

  \item For all $j' \in Q_{FC}$, if $f_{S, w'''}(j') = v(s)$ and $s \leq s'$, then there is a $j'' \in Q_{FC}$ such that $j' \leq_{FC} j''$ and $f_{S, w'''}(j'') = v(s')$.

    Assume $f_{S, w'''}(j') = v(s)$ and $s \leq s'$ for some $j' \in Q_{FC}$. From the definition of $f_{S, w'''}$ we have $f_{S,w}(j') = v(s'')$ for some $s'' R s$. From $F_{2}$ and witness uniqueness we conclude that, since $s'' R s$ and $s \leq s'$, there must be some unique $s'''$ such that $s'' \leq s'''$ and $s''' R s'$. Since $s'' \leq s'''$ and $f_{S, w}(j') = v(s'')$ we conclude by Clause 2 of Definition \ref{def:interpretationfunction} that there is some $j'' \in Q_{FC}$ such that $j' \leq_{FC} j''$ and $f_{S,w}(j'') = v(s''')$. Since $j \leq_{FC} j'$ and $j' \leq_{FC} j''$ by transitivity of $\leq_{FC}$ we conclude $j \leq_{FC} j''$. Since $f_{S,w}(j'') = v(s''')$ then $f_{S,w'''}(j'') = v(s'''')$, where $s''''$ is the unique object with $s''' R s''''$ and $w''' \leq s''''$. Since  $f_{S, w'''}(j') = v(s)$ and $f_{S,w'''}$ satisfies Clause 2 of Definition \ref{def:interpretationfunction} we conclude $w''' \leq s$, and since $s \leq s'$ by transitivity of $\leq$ we conclude $w''' \leq s'$. So we have $w''' \leq s'$, $s''' R s'$, $w''' \leq s''''$ and $s''' R s''''$, so from uniqueness of $s''''$ we conclude $s' = s''''$. Hence since  $f_{S,w'''}(j'') = v(s'''')$ we have $f_{S,w'''}(j'') = v(s')$, so $j''$ is the desired witness.

%Since $w'' \leq s'$, $s''' R s'$ and $s''''$ is unique we conclude $s' = s''''$, so $f_{B,w''}(j'') = v(s')$, and since $j' \leq_{PC} j''$ by Definition \ref{def:PCj} also $j' \leq_{PC^j} j''$, so $j''$ is the desired witness.

\item For all $j',j'' \in Q_{FC}$, if $f_{S, w'''}(j') = v(s)$, $f_{S, w'''}(j'') = v(s')$ and $j' \leq_{FC} j''$ then $s \leq s'$.

  Pick any two objects $j', j'' \in Q_{FC}$ with $f_{S, w'''}(j') = v(s)$, $f_{S, w'''}(j'') = v(s')$ and $j' \leq_{FC} j''$. By the definition of $f_{S,w'''}$ we have $f_{S,w}(j') = v(s'')$ for $s''$ with $s''Rs$ and $f_{S,w}(j'') = v(s''')$ for $s'''$ with $s'''Rs'$. Due to how $f_{S,w'''}$ is defined we also have that $s'$ is the unique object such that $w''' \leq s'$ and $s''' R s'$. Since $f_{S,w}(j') = v(s'')$ and $f_{S,w}(j'') = v(s''')$ we conclude $s'' \leq s'''$ by Clause 4 of Definition \ref{def:interpretationfunction}. Since $s'' \leq s'''$ and $s''Rs$, by $F_{1}$ and witness uniqueness there must be a unique $s''''$ such that $s \leq s''''$ and $s'''Rs''''$. Since $f_{S, w'''}(j') = v(s)$ and $f_{S,w'''}$ satisfies Clause 2 of Definition \ref{def:interpretationfunction} we conclude $w''' \leq s$, so since $s \leq s''''$ we conclude $w''' \leq s''''$ by transitivity of $\leq$. Since $w''' \leq s'$, $s'''Rs'$, $w''' \leq s''''$ and $s'''Rs'''''$ and the object $s'$ is unique we conclude $s' = s''''$, hence since $s \leq s''''$ we conclude $s \leq s'$.

\end{enumerate}

\end{proof}

\item \textbf{Lemma \ref{lemma:proofofequalityforworldshomogeneous}}

    \begin{proof}
        We prove the result by induction on the complexity of $A$. All cases except the ones with $\Gamma = \varnothing$ and either $A = \Box B$ or $A = \Diamond B$ are analogous to the ones of Lemma \ref{lemma:proofofequalityforworlds}, so we deal only with those.

        \begin{enumerate}

            \item $A = \Box B$. Assume $ \vDash_{j}^{M^S, K} \Box B$. Then $K \succ K'$ implies  $\vDash^{M^S,K'}_{j} B$. Pick any $w''$ with $w'Rw''$. By Lemma \ref{lemma:modalexistencehomogeneous}, there is a $K'' \in F$ with $K'' = \langle Q_{FC}, \leq_{FC}, f_{S,w'''} \rangle$ such that $f_{S,w'''}(j) = v(w'')$ and, for any object $j'' \in Q_{FC}$, $f_{S, w}(j'') = v(s)$ and $f_{S, w}(j'') = v(s')$ implies $sRs'$. By the definition of $\succ$ we have $K \succ K''$, so $\vDash^{M^S,K''}_{j} B$. Since $f_{S,w'''}(j) = v(w'')$ the induction hypothesis yields $\Vdash^{S}_{w''} A$, so by arbitrariness of $w''$ we conclude $\Vdash^{S}_{w'} \Box A$. For the converse, assume $\Vdash^{S}_{w'} \Box A$. Then $w'Rw''$ implies $\Vdash^{S}_{w''} A$. Pick any $K'$ with $K \succ K'$. Then $K' = \langle Q_{FC}, \leq_{FC}, f_{S,w'''} \rangle$ such that, for any $j'' \in Q_{FC}$, $f_{S, w}(j'') = v(s)$ and $f_{S, w}(j'') = v(s')$ implies $sRs'$. In particular, since $f_{S,w}(j) = v(w')$ we have $f_{S,w'''}(w'''')$ for $w'Rw''''$. Since $wRw''''$ we have $\Vdash^{S}_{w''''} A$. Since $f_{S,w'''}(j) = v(w'''')$ the induction hypothesis yields $ \vDash_{j}^{M^S, K'} B$, hence since $K'$ was arbitrary we conclude $ \vDash_{j}^{M^S, K} \Box B$.

            \item $A = \Diamond B$. Assume $\vDash_{j}^{M^S, K} \Diamond B$. Then there is a $K \succ K'$ with $ \vDash_{j}^{M^S, K'} B$. Without loss of generality, let $K' = \{Q_{FC}, \leq_{FC}, f_{S,w'''}\}$ and $f_{S,w'''}(j) = v(w'')$. Then since $f_{S,w}(j) = v(w')$ from the definition of $\succ$ we have $w' R w''$. Since $ \vDash_{j}^{M^S, K'} B$ and $f_{S,w'''}(j) = v(w'')$ the induction hypothesis yields $\Vdash^{S}_{w''} B$, hence since $w'Rw''$ we conclude $\Vdash^{S}_{w'} \Diamond B$. For the converse, assume $\Vdash^{S}_{w'} \Diamond B$. Then there is a $w''$ with $w' R w''$ and $\Vdash^{S}_{w''} B$. Since $w'Rw''$ and $f_{S,w}(j) = v(w')$, by Lemma \ref{lemma:modalexistencehomogeneous} we conclude that there is a $K' = \langle Q_{FC}, \leq_{FC}, f_{S,w'''} \rangle$ such that $f_{S,w'''}(j) = v(w'')$ and, if $j' \in Q_{FC}$, $f_{S,w}(j') = v(s)$ and $f_{S,w'''}(j') = v(s')$, then $s R s'$. By the definition of $\succ$ we have $K\succ K'$. Since $f_{S,w'''}(j) = v(w'')$ and and $\Vdash^{S}_{w''} B$ the induction hypothesis yields $ \vDash_{j}^{M^S, K'} B$, hence since $K \succ K'$ we conclude $\vDash_{j}^{M^S, K} \Diamond B$.

                    \end{enumerate}

    \end{proof}

\end{itemize}

\bibliography{references} 
\bibliographystyle{asl}

\end{document}